\documentclass[twoside]{article}
\usepackage{qic}
\usepackage{algorithm}
\usepackage{algorithmic}
\usepackage{amsmath}
\usepackage{amsthm}
\usepackage{amssymb}
\usepackage{hyperref}
\usepackage{bm}
\usepackage[title]{appendix}

\hypersetup{
    colorlinks=true,
    linkcolor=blue,
    urlcolor=cyan,
    citecolor=green
}

\usepackage{tikz}
\usepackage{lipsum}
\usepackage{float}
\usepackage{physics}

\newtheorem{theorem}{Theorem}[subsection]
\newtheorem{lemma}[theorem]{Lemma}

\newtheorem{proposition}[theorem]{Proposition}

\newtheorem{definition}[theorem]{Definition}
\newtheorem{example}[theorem]{Example}

\newtheorem{remark}{Remark}[subsection]

\numberwithin{equation}{section}

\begin{document}
\setlength{\textheight}{8.0truein}    %FOR 2ND PAGE ONWARDS

\runninghead{A Probabilistic Representation for Multi-State Discrete-time Quantum Walks}
            {Hoang Vu}

\normalsize\textlineskip
\thispagestyle{empty}

\vspace*{0.88truein}

%\alphfootnote

\fpage{1}

\centerline{\bf
%%%%%%%%%%%%%%%%%%%%%
%Put in titiles here
%%%%%%%%%%%%%%%%%%%%%
A Probabilistic Representation for Multi-State Discrete-time Quantum Walks}
\vspace*{0.035truein}
\vspace*{0.37truein}
\centerline{\footnotesize
%%%%%%%%%%%%%%%%%%%%%%%%%%%%%%%%%%%%
%put authors' name and address here
%%%%%%%%%%%%%%%%%%%%%%%%%%%%%%%%%%%%
Hoang Vu}
\vspace*{0.015truein}
\centerline{\footnotesize\it Department of Statistics and Applied Probability, University of California, Santa Barbara}
\baselineskip=10pt
\centerline{\footnotesize\it Santa Barbara, California 93106, United States}
\vspace*{10pt}

\vspace*{0.21truein}

\abstracts{
Building upon the pioneering framework of Vu (2026), we construct a probabilistic representation for three-state discrete-time quantum walks on integer lattices and validate it through empirical examples. Furthermore, we establish that this representation converges to the continuum solution of multi-state Dirac partial differential equations. Broadly, our findings demonstrate that this probabilistic paradigm serves as a robust alternative for simulating higher-dimensional quantum walks, opening new theoretical avenues to analyze quantum dynamics using classical stochastic processes.
}{}{}

\vspace*{10pt}

\vspace*{3pt}

\vspace*{1pt}\textlineskip   

\section{Introduction}        
Quantum walks (QWs) represent a foundational generalization of classical random walks, serving as an essential paradigm for quantum algorithm design and information processing. Divided broadly into discrete-time (coined) and continuous-time formulations, QWs have undergone extensive mathematical study since the pioneering work of Gudder \cite{gudder} and the subsequent introduction of quantum lattice gas automata by Meyer \cite{meyer}. The signature ballistic propagation of discrete-time Hadamard walks—initially characterized by Nayak and Vishwanath \cite{nayak} alongside Ambainis et al. \cite{amb}—stands in sharp contrast to classical diffusive behavior governed by the Central Limit Theorem. Konno \cite{konno, konno2} formally captured this non-classical dynamics through a weak limit theorem on one-dimensional lattices, a result later generalized by Grimmett et al. \cite{grimmett}. Nevertheless, extending such limit theorems to multi-dimensional manifolds remains a persistent and unresolved analytical challenge.\\

Traditionally, the study of QWs has been dominated by combinatorial techniques \cite{konno} and functional-analytic Fourier methods \cite{grimmett}, leaving purely probabilistic frameworks underutilized. This scarcity stems from the intrinsic nature of quantum walks, which are driven by deterministic, unitary evolution rather than true stochastic processes. Recently, however, emerging research \cite{konno3, mon, yama} indicates that examining QWs through a probabilistic lens reveals deep, hidden structural symmetries connecting quantum and classical dynamics. Formulating a probabilistic representation enables us to expose these underlying relationships and apply standard classical stochastic tools directly to quantum architectures.\\

In this paper, we establish a probabilistic representation for a three-state quantum walk on an integer lattice, allowing the system to be simulated via Monte Carlo schemes using purely classical stochastic processes. This representation traces its roots back to Molchanov's formula (see, e.g., \cite{carmona}, and \cite{vu} for a review), which predates the formal inception of quantum walks. While deriving a representation for a two-state quantum walk on a one-dimensional lattice from Molchanov's formula is relatively direct (see e.g. \cite{vu}), extending this construction to multi-state systems presents non-trivial hurdles. In Section \ref{sec2}, we review the existing two-state quantum walk framework originally proposed and empirically validated by Vu (2026) \cite{vu}. In Section \ref{sec3}, we demonstrate how to extend Vu's framework to three-state quantum walks, specifically addressing the technical difficulties introduced by interference effects within the coin space.\\

Lastly, Section \ref{sec4} presents our probabilistic formulation as a solver for multi-state systems of Dirac partial differential equations (PDEs). This result aligns with the theoretical finding by Maeda \& Suzuki (2020) \cite{mae}, which established that discrete-time quantum walks converge to solutions of Dirac PDEs in the continuum limit. By reinterpreting quantum walks through a stochastic lens, we offer a fresh perspective for analyzing high-dimensional discrete-time dynamics, providing a scalable foundation for Monte Carlo algorithms and paving the way for the development of multi-dimensional Dirac PDE solvers in future research.

\section{A Probabilistic Representation for Two-State Discrete-time Quantum Walk}{\label{sec2}}
\noindent

In this section, we will review the result obtained by Vu (2026)\cite{vu} for the two-state case. Let us define the two-state discrete-time quantum walk via the Hilbert space $\mathcal H$ such that 
\begin{align*}
    \mathcal H=\ell^2(\mathbb Z,\mathbb C^2)=\bigg\{\Psi:\mathbb Z\rightarrow \mathbb C^2\bigg |\sum_{x\in \mathbb Z}||\Psi(x)||^2_{\mathbb C^2}<\infty\bigg\}.
\end{align*}

Let $\mathcal{L}(\mathcal{H})$ denote the Banach space of bounded linear operators on $\mathcal{H}$, and let $\mathcal{U}(\mathcal{H}) \subset \mathcal{L}(\mathcal{H})$ be the subgroup of unitary operators. We denote the standard orthonormal basis of the coin space $\mathbb{C}^2$ by 
\begin{align*}
    \ket{-1}:=\begin{pmatrix}1\\0\end{pmatrix};\quad \quad 
    \ket{1}:=\begin{pmatrix}0\\1\end{pmatrix}.
\end{align*}

The quantum walker's state space consists two components: a coin state that control the walker direction, and a position state. The position space denoted by $\ell^2(\mathbb Z)=\text{Span}\{|x\rangle, x\in\mathbb Z\}$, and the coin space denoted by $\ell^2(\mathbb C^2)=\text{Span}\{\ket{y}, y=\pm 1\}$. The initial state of the quantum walker is given by
\begin{align}
    \ket{\Psi_0}=\ket{0}\otimes\Big(\alpha\ket{1}+\beta\ket{-1}\Big), \label{1122}
\end{align}
where $\alpha\in \mathbb C$, $\beta\in \mathbb C$, and $|\alpha|^2+|\beta|^2=1$ are the probability amplitudes corresponding to the coin state $\ket{1}$ and $\ket{-1}$ respectively at position $x=0$ at time $t=0$.\\

Next, we define a shift operator $S$ that move the quantum walker along the lattice as 
\begin{align}
        S\ket{x}\otimes\ket{y}=\ket{x+y}\otimes\ket{y}.
\end{align}

Then, the discrete-time quantum walk is defined as follows

\vspace*{12pt}
\noindent
\begin{definition} 
      A random quantum walk $\{Q_t\}$ under the Hilbert space $\mathcal H=\ell^2(\mathbb Z) \otimes \ell^2(\mathbb C^2)$, is determined by the unitary evolution operator $U\in \mathcal L(\mathcal H)$
    \begin{align}
        U=S\cdot \Big(\sum_{x\in \mathbb Z}|x\rangle\langle x|\otimes C(x)\Big),
    \end{align}
    where $C\in \mathcal U(\ell^2(\mathbb C^2))$ is the quantum coin. That is, $\{Q_t\}$ is a random variable defined by the probability distribution of the quantum state $\Psi_t$ after applying $t$-times the operator $U$ to the initial state $\Psi_0$.
\end{definition}
\vspace*{12pt}
\noindent

Note that any coin matrix $C\in \mathcal U(\ell^2(\mathbb C^2))$ can be written in the following form via the Euler angle decomposition (see e.g. \cite{todd},\cite{gre})
\begin{align}
C=e^{i\lambda_0}e^{i\lambda_1\sigma_3}e^{i\lambda_2\sigma_2}e^{i\lambda_3\sigma_3},\label{eq33}
\end{align}
where $\lambda_j\in (0,2\pi), j=0,1,2,3$; $\sigma_2$, and $\sigma_3$ are Pauli matrix $Y$ and $Z$ respectively , and are defined as follows
\begin{align*}
    \sigma_2=\begin{pmatrix}
        0 &-i\\
        i & 0
    \end{pmatrix}, \quad \quad \sigma_3=\begin{pmatrix}
        1 &0 \\
        0 &-1
    \end{pmatrix}.
\end{align*}

We introduce the following classical process to formulate our probabilistic representation:

\vspace*{12pt}
\noindent
\begin{definition}{\label{def4}}
    Let $N_1,N_2,...N_n$ be i.i.d Poisson random variables with parameter $\lambda \in (0,2\pi)$, we have
    \begin{align*}
    S_0&=0, \quad \quad S_n=\sum_{j=1}^nN_j  \quad \quad (n\geq 1),\\
    Y_0&=y, \quad \quad Y_n=(-1)^{S_n}\bigg(Y_0+\frac{a_{0,c}(Y_0)}{2}\bigg)-\frac{a_{0,c}(Y_0)}{2}(-1)^{S_n(c+1)} \quad \quad (n\geq 1),\\
    X_0&=x, \quad \quad X_n=X_{n-1}-Y_{n-1}=X_0-\sum_{j=0}^{n-1}Y_j \quad \quad (n\geq 1),
    \end{align*}
    where $a_{0,c}(Y_0)$ is a deterministic function of $y$ and $c$, and $c$ is a given fixed constant.
\end{definition}

\vspace*{12pt}
\noindent
\begin{remark}
    The defintion of $Y_n$ could be simpler here, but to keep it consistently with future research on high dimensional quantum walks, we insist to keep it in such a form.
\end{remark}
\vspace*{12pt}
\noindent

Now, consider the general coin in Equation \eqref{eq33}, $C=e^{i\lambda_0}e^{i\lambda_1\sigma_3}e^{i\lambda_2\sigma_2}e^{i\lambda_3\sigma_3}$, we have

\vspace*{12pt}
\noindent
\begin{lemma}{\label{lem8}}
    The probability amplitude evolution of a discrete-time quantum walk driven by the homogeneous coin $C=e^{i\lambda_0}e^{i\lambda_1\sigma_3}e^{i\lambda_2\sigma_2}e^{i\lambda_3\sigma_3}$ follows
         \begin{align}
         \Psi_n(x,y)=\sum_{k_1,k_2,...,k_n\in \mathbb N}e^{i\lambda_0n}e^{i\lambda_1\sum_{j=0}^{n-1}y_j}e^{i\lambda_3 \sum_{j=1}^{n}y_j}i^{\sum_{j=1}^nk_j+y_j\cdot\frac{1-(-1)^{k_j}}{2}}\frac{\lambda_2^{\sum_{j=1}^nk_j}}{k_1!k_2!...k_n!}\Psi_0(x_n,y_n),\label{eq8}
      \end{align}  
    where $x_n:=x_0-\sum_{j=0}^{n-1}y_j$, $y_n:=(-1)^{k_n}y_{n-1}$ for $n\geq 1$ with $(x_0,y_0)=(x,y)$.  
\end{lemma}

\vspace*{12pt}
\noindent
\begin{proof}
    See Vu(2026) \cite{vu}
\end{proof}
\vspace*{12pt}
\noindent

    This leads to the following representation theorem
    
\vspace*{12pt}
\noindent
\begin{theorem}{\label{theo8}}
    A discrete-time quantum walk driven by the homogeneous coin $C=e^{i\lambda_1\sigma_3}e^{i\lambda_2\sigma_2}e^{i\lambda_3\sigma_3}$ has the following probabilistic representation
    \begin{align}
        \Psi_n(x,y)=e^{n(\lambda_2+i\lambda_0)}\mathbb E\Big [i^{S_n+Y_0\cdot\frac{1-(-1)^{S_n}}{2}}e^{i\lambda_1(X_0-X_n)}e^{i\lambda_3(X_0-X_n+Y_n)}\Psi_0(X_n,Y_n)\Big],\label{eq9}
    \end{align}
    for $(x,y,n)\in\mathbb Z\times\{\pm1\}\times\mathbb N_0$, with  $\Psi_0(.,.)$ is defined by Equation \eqref{eq00}, and the classical processes $S_n$, $Y_n$, and $X_n$ are defined in Definition \ref{def4} with $c=0$.
\end{theorem}

\vspace*{12pt}
\noindent
\begin{proof}
    See Vu(2026) \cite{vu}
\end{proof}  

\vspace*{12pt}
\noindent

\begin{example}\label{ex1}
    Consider the Hadamard walk with the coin matrix 
    $$H=\frac{1}{\sqrt{2}}\begin{pmatrix}
        1&1\\1&-1
    \end{pmatrix},$$ which can also be written in the form
    \begin{align*}
        H=e^{i\frac{\pi}{2}\sigma_3}e^{i\frac{\pi}{4}\sigma_2}.
    \end{align*}
    According to Theorem \ref{theo8}, its probabilistic representation is
    \begin{align}
         \Psi_n(x,y)=e^{\frac{n\pi}{4}}\mathbb E\Big [i^{S_n+Y_0\cdot\frac{1-(-1)^{S_n}}{2}}e^{i\frac{\pi}{2}(X_0-X_n)}\Psi_0(X_n,Y_n)\Big].
    \end{align}
\end{example}

A general form of the initial state of the quantum walker is given by
\begin{align*}
    \ket{\Psi_0}=\ket{0}\otimes\Big(\alpha\ket{1}+\beta\ket{-1}\Big),
\end{align*}
where $\alpha\in \mathbb C$,$\beta\in \mathbb C$, and $|\alpha|^2+|\beta|^2=1$ are the probability amplitudes corresponding to the coin state $\ket{1}$ and $\ket{-1}$ respectively at position $x=0$ at time $t=0$. Hence, we can define the functional form of $\Psi_0(.,.)$ inside the expectation in Equation \eqref{eq9} by 
\begin{align}
    \Psi_0(x,y):=\mathbb I_{x=0}\Big(\alpha\cdot\mathbb I_{y=1}+\beta\cdot \mathbb I_{y=-1}\Big).\label{eq00}
\end{align}

From here, we can even rewrite Equation \eqref{eq9} in a more compact form
\begin{align}
     \Psi_n(x,y)=e^{n(\lambda_2+i\lambda_0)}\mathbb E\Big [i^{S_n+y\cdot\frac{1-(-1)^{S_n}}{2}}e^{i\lambda_1 x}e^{i\lambda_3(x+y(-1)^{S_n})}\Psi_0(X_n,Y_n)\Big].
\end{align}

 Note that, the initial state of the Hadamard walk is given by
$$
    \ket{\Psi_0}=\ket{0}\otimes\Big(\frac{1}{\sqrt{2}}\ket{1}+i\frac{1}{\sqrt{2}}\ket{-1}\Big).
$$ 

Vu(2026) \cite{vu} obtained the numerical simulation results shown in Figure \ref{fig1}, which confirm the validity of the probabilistic formula.

\begin{figure}[H]
\vspace*{13pt}
\centerline{\epsfig{file=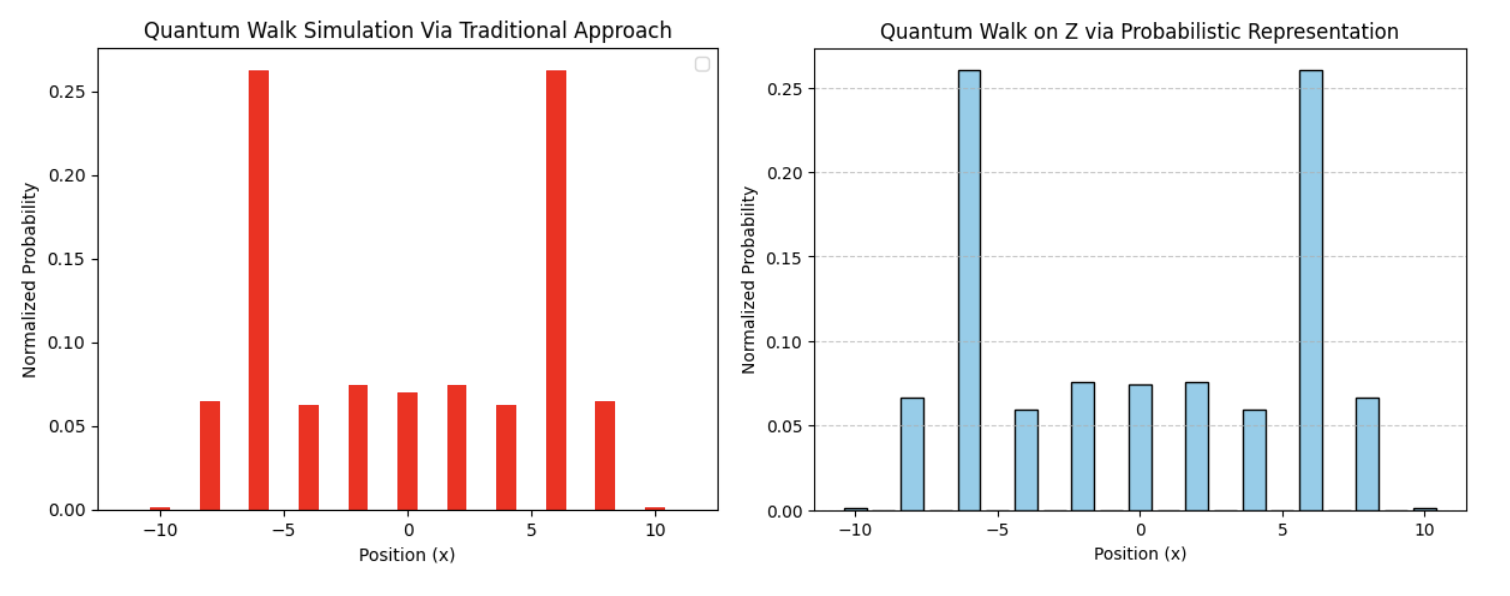, width=15cm}} %100 percent
\vspace*{13pt}
\fcaption{\label{fig1} The Hadamard walk's probability distribution for $n=10$, $\alpha=\frac{1}{\sqrt{2}}$, and $\beta=\frac{1}{\sqrt{2}}i$ with the left bar chart illustrating the benchmark method, and the right bar chart illustrating the probabilistic method with the number of iteration $M=5\times10^9$,$\lambda_0=0$, $\lambda_1=\frac{\pi}{2}$, $\lambda_2=\frac{\pi}{4}$, and $\lambda_3=0$.  }
\end{figure}

\section{A Probabilistic Representation for Multi-State Discrete-time Quantum Walks}{\label{sec3}}
In this section, we will extend the exisiting two-state model of Vu \cite{vu} to the three-state case. Let us define the three-state discrete-time quantum walk via the Hilbert space $\mathcal H$ such that 
\begin{align*}
    \mathcal H=\ell^2(\mathbb Z,\mathbb C^3)=\bigg\{\Psi:\mathbb Z\rightarrow \mathbb C^3\bigg |\sum_{x\in \mathbb Z}||\Psi(x)||^2_{\mathbb C^3}<\infty\bigg\}.
\end{align*}

 We denote the standard orthonormal basis of the coin space $\mathbb{C}^3$ by 
\begin{align*}
    \ket{-1}:=\begin{pmatrix}1\\0\\0\end{pmatrix};\quad \quad 
    \ket{0}:=\begin{pmatrix}0\\1\\0\end{pmatrix}; \quad \quad \ket{1}:=\begin{pmatrix}0\\0\\1\end{pmatrix}.
\end{align*}

The position space denoted by $\ell^2(\mathbb Z)=\text{Span}\{|x\rangle, x\in\mathbb Z\}$, and the coin space denoted by $\ell^2(\mathbb C^3)=\text{Span}\{\ket{y}, y=-1,0,1\}$. The initial state of the quantum walker is given by
\begin{align}
    \ket{\Psi_0}=\ket{0}\otimes\Big(\alpha\ket{1}+\beta\ket{0}+\gamma\ket{-1}\Big), \label{1122}
\end{align}
where $\alpha\in \mathbb C$, $\beta\in \mathbb C$, and $|\alpha|^2+|\beta|^2+|\gamma|^2=1$ are the probability amplitudes corresponding to the coin state $\ket{1}$,$\ket{0}$ and $\ket{-1}$ respectively at position $x=0$ at time $t=0$.\\

Next, we define a shift operator $S$ that move the quantum walker along the lattice as 
\begin{align}
        S\ket{x}\otimes\ket{y}=\ket{x+y}\otimes\ket{y}.
\end{align}

Then, the discrete-time quantum walk is defined as follows

\vspace*{12pt}
\noindent
\begin{definition}\label{def3} 
      A random quantum walk $\{Q_t\}$ under the Hilbert space $\mathcal H=\ell^2(\mathbb Z) \otimes \ell^2(\mathbb C^3)$, is determined by the unitary evolution operator $U\in \mathcal L(\mathcal H)$
    \begin{align}
        U=S\cdot \Big(\sum_{x\in \mathbb Z}|x\rangle\langle x|\otimes C(x)\Big),
    \end{align}
    where $C\in \mathcal U(\ell^2(\mathbb C^3))$ is the quantum coin. That is, $\{Q_t\}$ is a random variable defined by the probability distribution of the quantum state $\Psi_t$ after applying $t$-times the operator $U$ to the initial state $\Psi_0$.
\end{definition}
\vspace*{12pt}
\noindent

Note that any coin matrix $C\in \mathcal U(\ell^2(\mathbb C^3))$ can be written in the following form via the Euler angle decomposition(see e.g. \cite{todd},\cite{gre})
\begin{align}
    C=e^{i\lambda_0}e^{i\lambda_1g_3}e^{i\lambda_2g_2}e^{i\lambda_3g_3}e^{i\lambda_4g_5}e^{i\lambda_5g_3}e^{i\lambda_6g_2}e^{i\lambda_7g_3}e^{i\lambda_8g_8},\label{eq0}
\end{align}
where $\lambda_j\in (0,2\pi), j=1,2,...,8$, and $g_j, j=1,2,....,8$ are Gell-Mann matrices defined as follows
\begin{align*}
    g_1&=\begin{pmatrix}
        0&1&0\\1&0&0\\0&0&0
    \end{pmatrix},\quad \quad g_2=\begin{pmatrix}
        0&-i&0\\i&0&0\\0&0&0
    \end{pmatrix},\\ g_3&=\begin{pmatrix}
        1&0&0\\0&-1&0\\0&0&0
    \end{pmatrix},\quad \quad
    g_4=\begin{pmatrix}
        0&0&1\\0&0&0\\1&0&0
    \end{pmatrix},\\ 
    g_5&=\begin{pmatrix}
        0&0&-i\\0&0&0\\i&0&0
    \end{pmatrix},\quad \quad g_6=\begin{pmatrix}
        0&0&0\\0&0&1\\0&1&0
    \end{pmatrix},\\
    g_7&=\begin{pmatrix}
        0&0&0\\0&0&-i\\0&i&0
    \end{pmatrix},\quad \quad g_8=\frac{1}{\sqrt{3}}\begin{pmatrix}
        1&0&0\\0&1&0\\0&0&-2
    \end{pmatrix}.
\end{align*}

A popular example of the three-state quantum walk on an integer line is the Grover walk driven by the following Grover coin
\begin{align}
    C_{\text{grover}}=\frac{1}{3}\begin{pmatrix}
        -1&2&2\\
        2&-1&2\\
        2&2&-1
    \end{pmatrix},\label{eqg}
\end{align}
which carries an important feature called localization that does not exist in a two state Hadamard walk on a line. We can also write the Grover coin matrix via the Gell-Mann matrices stated above
\begin{align}
    C_{\text{grover}}= e^{i\frac{7\pi}{4}g_2} e^{i\arccos(-1/3)g_5} e^{i\frac{5\pi}{4}g_2}.\label{eqg}
\end{align}

The Euler angle decomposition motivates us to investigate the probabilistic representation of the homogeneous Gell-Mann coin, and here we only need to consider four Gell-Mann matrices, which are $g_2$,$g_3$,$g_5$, and $g_8$.

\subsection{A Formula for The Gell-Mann Homogenous Coin}
Let us first consider the coin $C=e^{i\lambda g_2}$, we have

\begin{lemma}{\label{lem2}}
    The probability amplitude evolution of a three-state quantum walk driven by the homogeneous coin $C=e^{i\lambda g_2}$ follows
         \begin{align}
        \Psi_n(x,y)=\sum_{k_1,k_2,...,k_n\in \mathbb N}i^{\sum_{j=1}^nk_j+a_1(k_j,y_j)}\frac{\lambda^{\sum_{j=1}^nk_j}}{k_1!k_2!...k_n!}a_0^{\sum_{j=1}^nk_j}(y_0)\Psi_0(x_n,y_n),\label{eq1}
      \end{align}  
    where $x_n:=x_0-\sum_{j=0}^{n-1}y_j$, $y_n:=(-1)^{k_n}y_{n-1}-\frac{1-(-1)^{k_n}}{2}(1-\mathbb I_{y_0=1})$ with $(x_0,y_0)=(x,y)$, $a_0(y_0)=1-\mathbb I_{y_0=1}$, and $a_{1,2}(k_j,y_j)=\frac{1-(-1)^{k_j}}{2}(2y_{j-1}+3)$.  
\end{lemma}

\begin{proof}
    First observe that 
      \begin{align*}
        U\ket{x}\ket{y}&=S\cdot(I\otimes C)\ket{x}\ket{y}\\
        &=S\ket{x}e^{i\lambda g_2}\ket{y}\\
        &=\sum_{k\in \mathbb N}S\ket{x}\frac{(i\lambda)^k}{k!}g_2^k\ket{y}\\
        &=\sum_{k\in\mathbb N}\frac{(i\lambda)^k}{k!}a^k_0(y)i^{a_{1,2}(k,y)}\ket{x+a_2(k,y)}\ket{a_2(k,y)}\\
        &=\sum_{k\in\mathbb N}i^{k+a_{1,2}(k,y)}\frac{\lambda^k}{k!}a_0^k(y)\ket{x+a_2(k,y)}\ket{a_2(k,y)},
    \end{align*}   
     where $a_0(y)=1-\mathbb I_{y=1}$, $a_{1,2}(k,y)=\frac{1-(-1)^k}{2}(2y+3)$, and $a_2(k,y)=(-1)^ky-\frac{1-(-1)^k}{2}(1-\mathbb I_{y=1})$.\\
     
    Now, for any state $\Psi$ of the walk, we have
     \begin{align*}
        U\Psi&=U\sum_{\substack{x\in\mathbb Z\\y\in\{0,\pm 1\}}}\Psi(x,y)\ket{x}\ket{y}\\
        &=\sum_{\substack{x\in\mathbb Z\\y\in\{0,\pm 1\}\\k\in \mathbb N}}\Psi(x,y)i^{k+a_{1,2}(k,y)}\frac{\lambda^k}{k!}a_0^k(y)\ket{x+a_2(k,y)}\ket{a_2(k,y)}\\
        &=\sum_{\substack{x\in\mathbb Z\\y\in\{0,\pm 1\}\\k\in \mathbb N}}i^{k+a_{1,2}(k,a_2(k,y))}\frac{\lambda^k}{k!}a_0^k(y)\Psi(x-y,a_2(k,y))\ket{x}\ket{y}
    \end{align*}   
   
   This implies that 
   \begin{align}
       (U\Psi)(x,y)=\sum_{k\in\mathbb N}i^{k+a_{1,2}(k,a_2(k,y))}\frac{\lambda^k}{k!}a_0^k(y)\Psi(x-y,a_2(k,y))
   \end{align}
  
    Hence, the evolution after $n-$steps yields the probability amplitude
    \begin{align}
       \Psi_n(x,y)=\sum_{k_1,k_2,...,k_n\in \mathbb N}i^{\sum_{j=1}^nk_j+a_{1,2}(k_j,y_{j-1})}\frac{\lambda^{\sum_{j=1}^nk_j}}{k_1!k_2!...k_n!}a_0^{\sum_{j=1}^nk_j}(y_0)\Psi(x-y_0-\sum_{j=1}^{n-1}a_2(k_j,y_j),a_2(k_n,y_n)),
   \end{align}
   where $a_2(k_j,y_j)=(-1)^{k_j}y_{j-1}-\frac{1-(-1)^{k_j}}{2}(1-\mathbb I_{y_0=1})$ for $j\geq 1$. This completes our proof. 
\end{proof}

We introduce the following classical process to formulate our probabilistic representation
\begin{definition}{\label{def1}}
    Let $N_1,N_2,...N_n$ be i.i.d Poisson random variables with parameter $\lambda \in (0,2\pi)$, we have
    \begin{align*}
    S_0&=0, \quad \quad S_n=\sum_{j=1}^nN_j  \quad \quad (n\geq 1),\\
    Y_0&=y, \quad \quad Y_n=(-1)^{S_n}\bigg(Y_0+\frac{a_{0,c}(Y_0)}{2}\bigg)-\frac{a_{0,c}(Y_0)}{2}(-1)^{S_n(c+1)} \quad \quad (n\geq 1),\\
    X_0&=x, \quad \quad X_n=X_{n-1}-Y_{n-1}=X_0-\sum_{j=0}^{n-1}Y_j \quad \quad (n\geq 1),
    \end{align*}
    where $a_{0,c}(Y_0)=1-\mathbb I_{Y_0=c}$, and $c$ is a given fixed constant.
\end{definition}

    This leads to the following representation theorem
\begin{theorem}
    A three-state quantum walk driven by the homogeneous coin $C=e^{i\lambda g_2}$ has the following probabilistic representation
    \begin{align}
        \Psi_n(x,y)=e^{n\lambda}\mathbb E\Big [i^{S_n+f_n}a_{0,1}^{S_n}(Y_0)\Psi_0(X_n,Y_n)\Big|(S_0,X_0,Y_0)=(0,x,y)\Big]\label{eqg2}
    \end{align}
    for $(x,y,n)\in\mathbb Z\times\{0,\pm1\}\times\mathbb N_0$, where $f_n:=\sum_{j=1}^n\frac{(1-(-1)^{N_j})(2Y_{j-1}+3)}{2}$, and $\Big(S.,X.,Y.,a_{.,.}(Y_0)\Big)$ were defined in Definition \ref{def1} with $c=1$.
\end{theorem}
\begin{proof}
    From Equation \eqref{eq1} in Lemma \ref{lem2}, apply the Poisson distribution, we have
    \begin{align*}
         \Psi_n(x_0,y_0)&=\sum_{k_1,...,k_n\in\mathbb N}i^{\sum_{j=1}^n k_j+\frac{(1-(-1)^{k_j})(2y_{j-1}+3)}{2}}\frac{\lambda^{\sum_{j=1}^nk_j}}{k_1!...k_n!}a_0^{\sum_{j=1}^nk_j}(y_0)\Psi_0(x_n,y_n)\\
         &=e^{n\lambda}\sum_{k_1,...,k_n\in\mathbb N}i^{\sum_{j=1}^n k_j+\frac{(1-(-1)^{k_j})(2y_{j-1}+3)}{2}}\frac{e^{-\lambda}\lambda^{k_1}...e^{-\lambda}\lambda^{k_n}}{k_1!...k_n!}a_0^{\sum_{j=1}^nk_j}(y_0)\Psi_0(x_n,y_n)\\
         &=e^{n\lambda}\mathbb E\Big [i^{S_n+f_n}a_{0,1}^{S_n}(Y_0)\Psi_0(X_n,Y_n)\Big|(S_0,X_0,Y_0)=(0,x,y)\Big],
    \end{align*}
    for $x_0=x$, $y_0=y$, and $f_n=\sum_{j=1}^n\frac{(1-(-1)^{k_j})(2y_{j-1}+3)}{2}$. This completes our proof.
\end{proof}

Next, let us consider the coin $C=e^{i\lambda g_3}$, we have
\begin{lemma}{\label{lem3}}
    The probability amplitude evolution of a three-state quantum walk driven by the homogeneous coin $C=e^{i\lambda g_3}$ follows
         \begin{align}
        \Psi_n(x,y)=\sum_{k_1,k_2,...,k_n\in \mathbb N}i^{\sum_{j=1}^nk_j}\frac{\lambda^{\sum_{j=1}^nk_j}}{k_1!k_2!...k_n!}a_0^{\sum_{j=1}^nk_j}(y_0)(-1)^{(y_0+1)\sum_{j=1}^nk_j}\Psi_0(x_n,y_n),\label{eq2}
      \end{align}  
    where $x_n:=x_0-ny_0$, $y_n:=y_{n-1}=...=y_0$ with $(x_0,y_0)=(x,y)$, and $a_0(y_0)=1-\mathbb I_{y_0=1}$.  
\end{lemma}

\begin{proof}
 See Appendix \ref{a31}.
\end{proof}

    This leads to the following representation theorem
\begin{theorem}
    A three-state quantum walk driven by the homogeneous coin $C=e^{i\lambda g_3}$ has the following probabilistic representation
    \begin{align}
        \Psi_n(x,y)=e^{n\lambda}\Psi_0(x-ny_0,y_0)\mathbb E\Big [i^{S_n}a_{0,1}^{S_n}(y_0)(-1)^{S_n(y_0+1)}\Big]
    \end{align}
    for $(x,y,n)\in\mathbb Z\times\{0,\pm1\}\times\mathbb N_0$, where  $\Big(S.,a_{.,.}(y_0)\Big)$ were defined in Definition \ref{def1} with $c=1$.
\end{theorem}
\begin{proof}
    See Appendix \ref{a32}.
\end{proof}

Next, let consider the coin $C=e^{i\lambda g_5}$, we have:
\begin{lemma}{\label{lem4}}
    The probability amplitude evolution of a three-state quantum walk driven by the homogeneous coin $C=e^{i\lambda g_5}$ follows
         \begin{align}
        \Psi_n(x,y)=\sum_{k_1,k_2,...,k_n\in \mathbb N}i^{\sum_{j=1}^nk_j+a_{1,5}(k_j,y_j)}\frac{\lambda^{\sum_{j=1}^nk_j}}{k_1!k_2!...k_n!}a_0^{\sum_{j=1}^nk_j}(y_0)\Psi_0(x_n,y_n),\label{eq3}
      \end{align}  
    where $x_n:=x_0-\sum_{j=0}^{n-1}y_j$, $y_n:=(-1)^{k_n}y_{n-1}$ with $(x_0,y_0)=(x,y)$, $a_0(y_0)=1-\mathbb I_{y_0=0}$, and $a_{1,5}(k_j,y_{j-1})=\frac{1-(-1)^{k_j}}{2}(y_{j-1}+2)$.  
\end{lemma}

\begin{proof}
   See Appendix \ref{a33}. 
\end{proof}

    This leads to the following representation theorem
\begin{theorem}
    A three-state quantum walk driven by the homogeneous coin $C=e^{i\lambda g_5}$ has the following probabilistic representation
    \begin{align}
        \Psi_n(x,y)=e^{n\lambda}\mathbb E\Big [i^{S_n+r_n}a_{0,0}^{S_n}(Y_0)\Psi_0(X_n,Y_n)\Big|(S_0,X_0,Y_0)=(0,x,y)\Big]
    \end{align}
    for $(x,y,n)\in\mathbb Z\times\{0,\pm1\}\times\mathbb N_0$, where $r_n=\sum_{j=1}^n\frac{(1-(-1)^{N_j})(Y_{j-1}+2)}{2}$, and $\Big(S.,X.,Y.,a_{.,.}(Y_0)\Big)$ were defined in Definition \ref{def1} with $c=0$.
\end{theorem}
\begin{proof}
    See Appendix \ref{a34}.
\end{proof}

Finally, let us consider the coin $C=e^{i\lambda g_8}$, we have
\begin{lemma}{\label{lem5}}
    The probability amplitude evolution of a three-state quantum walk driven by the homogeneous coin $C=e^{i\lambda g_8}$ follows
         \begin{align}
        \Psi_n(x,y)=\sum_{k_1,k_2,...,k_n\in \mathbb N}i^{\sum_{j=1}^nk_j}\frac{\lambda^{\sum_{j=1}^nk_j}}{k_1!k_2!...k_n!}b_0^{\sum_{j=1}^nk_j}(y_0)\Psi_0(x_n,y_n),\label{eq4}
      \end{align}  
    where $x_n:=x_0-ny_0$, $y_n:=y_{n-1}=...=y_0$ with $(x_0,y_0)=(x,y)$, and $b_0(y_0)=\frac{(-1)^{\mathbb I_{y_0=1}}(\mathbb I_{y_0=1}+1)}{\sqrt{3}}$.  
\end{lemma}

\begin{proof}
    See Appendix \ref{a35}.

\end{proof}

    This leads to the following representation theorem
\begin{theorem}
    A three-state quantum walk driven by the homogeneous coin $C=e^{i\lambda g_8}$ has the following probabilistic representation:
    \begin{align}
        \Psi_n(x,y)=e^{n\lambda}\Psi_0(x-ny_0,y_0)\mathbb E\Big [i^{S_n}b_{0}^{S_n}(y_0)\Big]
    \end{align}
    for $(x,y,n)\in\mathbb Z\times\{0,\pm1\}\times\mathbb N_0$, where  $S_n$ was defined in Definition \ref{def1}.
\end{theorem}
\begin{proof}
   See Appendix \ref{a36}.
\end{proof}
\subsection{A Formula for The General Coin}
Now, consider the general coin $C=e^{i\lambda_0}e^{i\lambda_1g_3}e^{i\lambda_2g_2}e^{i\lambda_3g_3}e^{i\lambda_4g_5}e^{i\lambda_5g_3}e^{i\lambda_6g_2}e^{i\lambda_7g_3}e^{i\lambda_8g_8}$. First, we define the following operator
\begin{align*}
    \mathcal{T}_2(y,k)&=(-1)^ky-\frac{1-(-1)^k}{2}a_{0,1}(y),\\
    \mathcal{T}_5(y,k)&=(-1)^ky,\\
    \theta_{78}&=\lambda_7\theta_3(y)+\lambda_8\theta_8(y),
\end{align*}
where $\theta_3(.)$ and $\theta_8(.)$ are eigenvalues of $g_3$ and $g_8$ respectively with the following values
\begin{align*}
    \theta_3(-1)=1,\quad \quad \theta_3(0)=-1, \quad \quad \theta_3(1)=0,
\end{align*}
\begin{align*}
    \theta_8(-1)=1/\sqrt{3}, \quad \quad  \theta_8(0)=1/\sqrt{3}, \quad \quad  \theta_8(1)=-2/\sqrt{3}.
\end{align*}

Similarly to the single Gell-Mann coin case, we have the following lemma
\begin{lemma}{\label{lem_gen}}
The probability amplitude evolution of a three-state quantum walk driven by the general coin $C$ defined in Equation \eqref{eq0} follows
\begin{align}
\Psi_n(x,y) = \sum_{\substack{k_1^{(1)},\ldots,k_n^{(1)}\in\mathbb N\\
k_1^{(2)},\ldots,k_n^{(2)}\in\mathbb N\\ k_1^{(3)},\ldots,k_n^{(3)}\in\mathbb N}}
\prod_{j=1}^n\big(\Phi_j\cdot\mathcal A_j\big)\;\Psi_0(x_n,y_n^{(3)}),
\end{align}
where, with $(x_0, y_0^{(3)}) := (x, y)$ and, for $j = 1, \dots, n$,
\[
    y_j^{(0)} := y_{j-1}^{(3)}, \quad y_j^{(1)} := \mathcal{T}_2(y_j^{(0)}, k_j^{(1)}), \quad y_j^{(2)} := \mathcal{T}_5(y_j^{(1)}, k_j^{(2)}), \quad y_j^{(3)} := \mathcal{T}_2(y_j^{(2)}, k_j^{(3)}),
\]
we have
\begin{align*}
\mathcal A_j = \;&\frac{(i\lambda_2)^{k_j^{(1)}}}
{k_j^{(1)}!}\,a_{0,1}^{k_j^{(1)}}\!\big(y_j^{(0)}\big)\,i^{\,a_{1,2}(k_j^{(1)},\,y_j^{(1)})}\\
\times\;&\frac{(i\lambda_4)^{k_j^{(2)}}}
{k_j^{(2)}!}\,a_{0,0}^{k_j^{(2)}}\!\big(y_j^{(1)}\big)\,i^{\,a_{1,5}(k_j^{(2)},\,y_j^{(2)})}\\
\times\;&\frac{(i\lambda_6)^{k_j^{(3)}}}
{k_j^{(3)}!}\,a_{0,1}^{k_j^{(3)}}\!\big(y_j^{(2)}\big)\,i^{\,a_{1,2}(k_j^{(3)},\,y_j^{(3)})},
\end{align*}
\begin{align*}
\Phi_j = \exp\!\Big[i\Big(\lambda_0+\theta_{78}\big(y_j^{(3)}\big) +
\lambda_5\theta_3\big(y_j^{(2)}\big) + \lambda_3\theta_3\big(y_j^{(1)}\big) +
\lambda_1\theta_3\big(y_j^{(0)}\big)\Big)\Big],
\end{align*}
and $x_n := x_0 - \sum_{j=0}^{n-1} y_j^{(3)}$.
\end{lemma}

\begin{proof}
    Apply sequentially each single Gell-Mann coin operators, and using Lemma \ref{lem2}, \ref{lem3}, \ref{lem4}, and \ref{lem5} gives the results. 
\end{proof}

To construct the probabilistic representation, we give the following definition

\begin{definition}{\label{def_gen}}
Let $\{N_j^{(1)}\}_{j=1}^n$, $\{N_j^{(2)}\}_{j=1}^n$, $\{N_j^{(3)}\}_{j=1}^n$ be independent Poisson random variables with
parameters $\lambda_2$, $\lambda_4$, $\lambda_6$ respectively. Define
\begin{align*}
Y_0^{(3)} &= y,\\
\text{for } j=1,\ldots,n:\quad Y_j^{(0)} &= Y_{j-1}^{(3)},\\
Y_j^{(1)} &= (-1)^{N_j^{(1)}}Y_j^{(0)} - \frac{1-(-1)^{N_j^{(1)}}}
{2}a_{0,1}\big(Y_j^{(0)}\big),\\
Y_j^{(2)} &= (-1)^{N_j^{(2)}}Y_j^{(1)},\\
Y_j^{(3)} &= (-1)^{N_j^{(3)}}Y_j^{(2)} - \frac{1-(-1)^{N_j^{(3)}}}
{2}a_{0,1}\big(Y_j^{(2)}\big),\\
X_0 &= x,\qquad X_n = X_{n-1}-Y_{n-1}^{(3)} = X_0 - \sum_{j=0}^{n-1}Y_j^{(3)}.
\end{align*}
\end{definition}

\begin{theorem}{\label{theo_gen}}
A three-state quantum walk driven by the general homogeneous coin C defined in Equation \ref{eq0} has the probabilistic representation
\begin{align}
\Psi_n(x,y) = e^{n(i\lambda_0+\lambda_2+\lambda_4+\lambda_6)}\,\mathbb
E\big[\Xi_n\cdot\Psi_0(X_n,Y_n^{(3)})\big],
\end{align}
where
\begin{align}
\Xi_n = \prod_{j=1}^n\Big(&i^{N_j^{(1)}+N_j^{(2)}+N_j^{(3)}}\nonumber\\
\times\;&i^{\,a_{1,2}(N_j^{(1)},\,Y_j^{(1)})\,+\,a_{1,5}(N_j^{(2)},\,Y_j^{(2)})\,+\,a_{1,2}
(N_j^{(3)},\,Y_j^{(3)})}\nonumber\\
\times\;&a_{0,1}^{N_j^{(1)}}\!\big(Y_j^{(0)}\big)\,a_{0,0}^{N_j^{(2)}}\!\big(Y_j^{(1)}\big)\,a
_{0,1}^{N_j^{(3)}}\!\big(Y_j^{(2)}\big)\nonumber\\
\times\;&\exp\!\Big[i\Big(\lambda_0+\theta_{78}\big(Y_j^{(3)}\big)+\lambda_5\theta_3\big(Y_j^{(2)}\big
)+\lambda_3\theta_3\big(Y_j^{(1)}\big)+\lambda_1\theta_3\big(Y_j^{(0)}\big)\Big)\Big]\Big).
\end{align}
\end{theorem}
\begin{proof}
    From Lemma \ref{lem_gen}, the three summations over $k_j^{(1)}, k_j^{(2)}, k_j^{(3)}$ at each step $j$ are the
Taylor expansions of $e^{i\lambda_2}, e^{i\lambda_4}, e^{i\lambda_6}$. For a generic term with parameter $\lambda$ and index $k$,

\[
\sum_k \frac{\lambda^k}{k!}(\cdots) = e^\lambda \sum_k \frac{e^{-\lambda}\lambda^k}{k!}(\cdots) = e^\lambda \mathbb E\big[(\cdots)\big|_{k\to N}\big].
\]

Applying this to all $3n$ summations extracts the prefactor $e^{n(\lambda_2+\lambda_4+\lambda_6)}$; substituting
$k_j^{(1)}, k_j^{(2)}, k_j^{(3)} \to N_j^{(1)}, N_j^{(2)}, N_j^{(3)}$ throughout $\mathcal A_j, \Phi_j$ and $y_j^{(\cdot)} \to Y_j^{(\cdot)}$ throughout
reproduces $\Xi_n$ and Definition \ref{def_gen}, completing the proof.
\end{proof}

\subsection{Empirical Analysis}
Consider again the three-state quantum walk driven by the coin matrix $C=e^{i\lambda g_2}$ with the initial state
\begin{align*}
     \ket{\Psi_0}=\ket{0}\otimes\Big(\frac{1}{\sqrt{2}}\ket{1}+0\ket{0}+i\frac{1}{\sqrt{2}}\ket{-1}\Big).
\end{align*}
   
\begin{figure}[H]
\vspace*{8pt}
\centerline{\epsfig{file=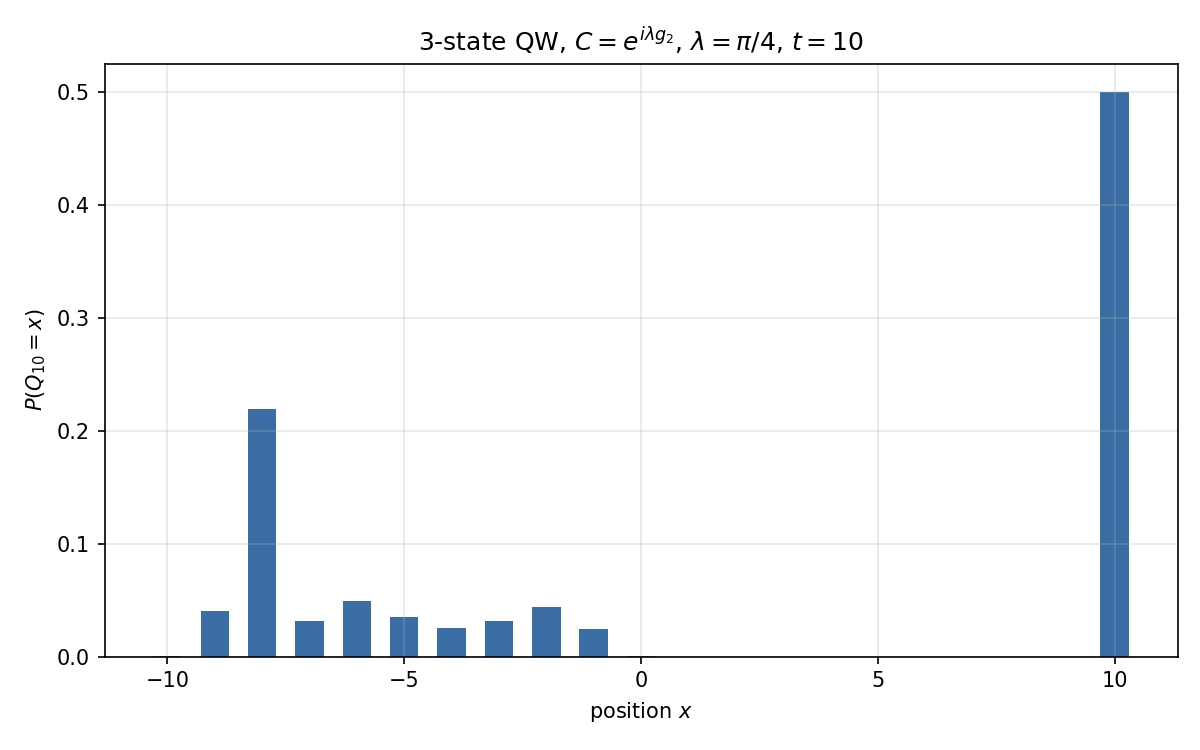, width=10cm}} %100 percent
\vspace*{8pt}
\fcaption{\label{fig2} The three-state quantum walk's probability distribution with the coin $C=e^{i\lambda g_2}$ for $n=10$, $\alpha=\frac{1}{\sqrt{2}}$, $\beta=0$ and $\gamma=\frac{1}{\sqrt{2}}i$,and $\lambda=\frac{\pi}{4}$.  }
\end{figure}

\begin{figure}[H]
\vspace*{13pt}
\centerline{\epsfig{file=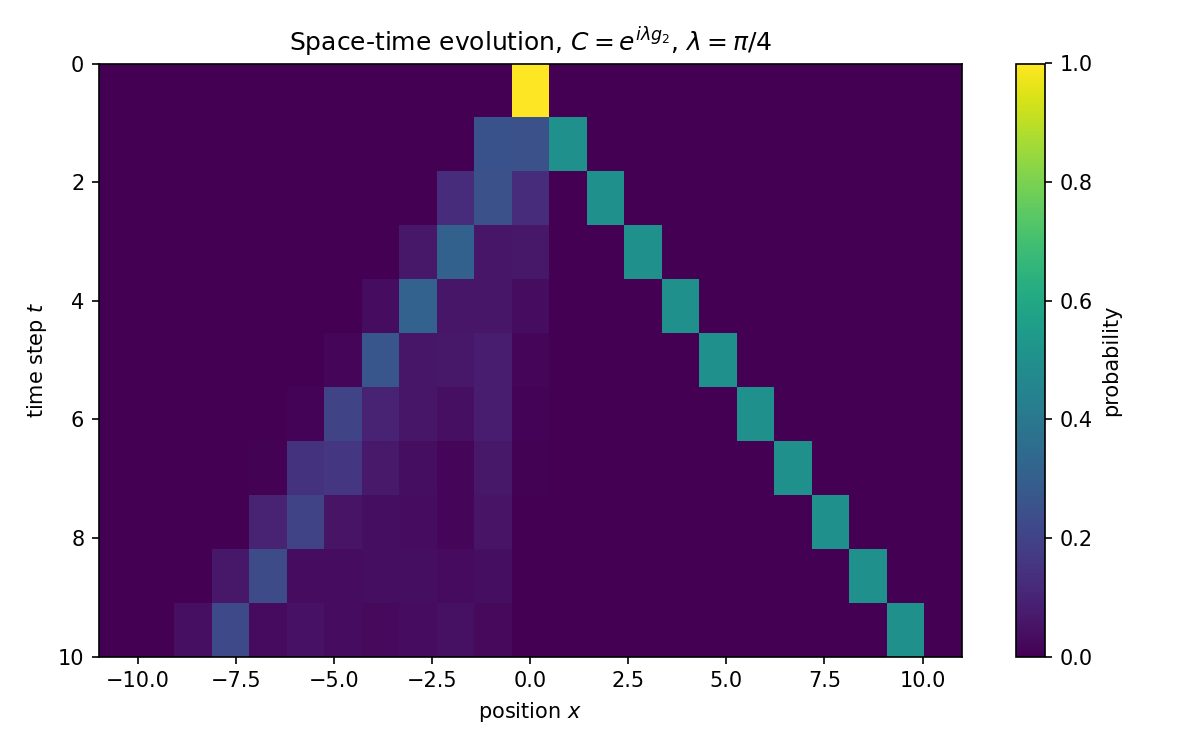, width=10cm}} %100 percent
\vspace*{13pt}
\fcaption{\label{fig3} The three-state quantum walk's spacetime behavior with the coin $C=e^{i\lambda g_2}$ for $n=10$, $\alpha=\frac{1}{\sqrt{2}}$, $\beta=0$ and $\gamma=\frac{1}{\sqrt{2}}i$, and $\lambda=\frac{\pi}{4}$.  }
\end{figure}

Figure \ref{fig2} and \ref{fig3} show the distribution of the walk via unitary evolution and its spacetime behavior.

\begin{figure}[H]
\vspace*{13pt}
\centerline{\epsfig{file=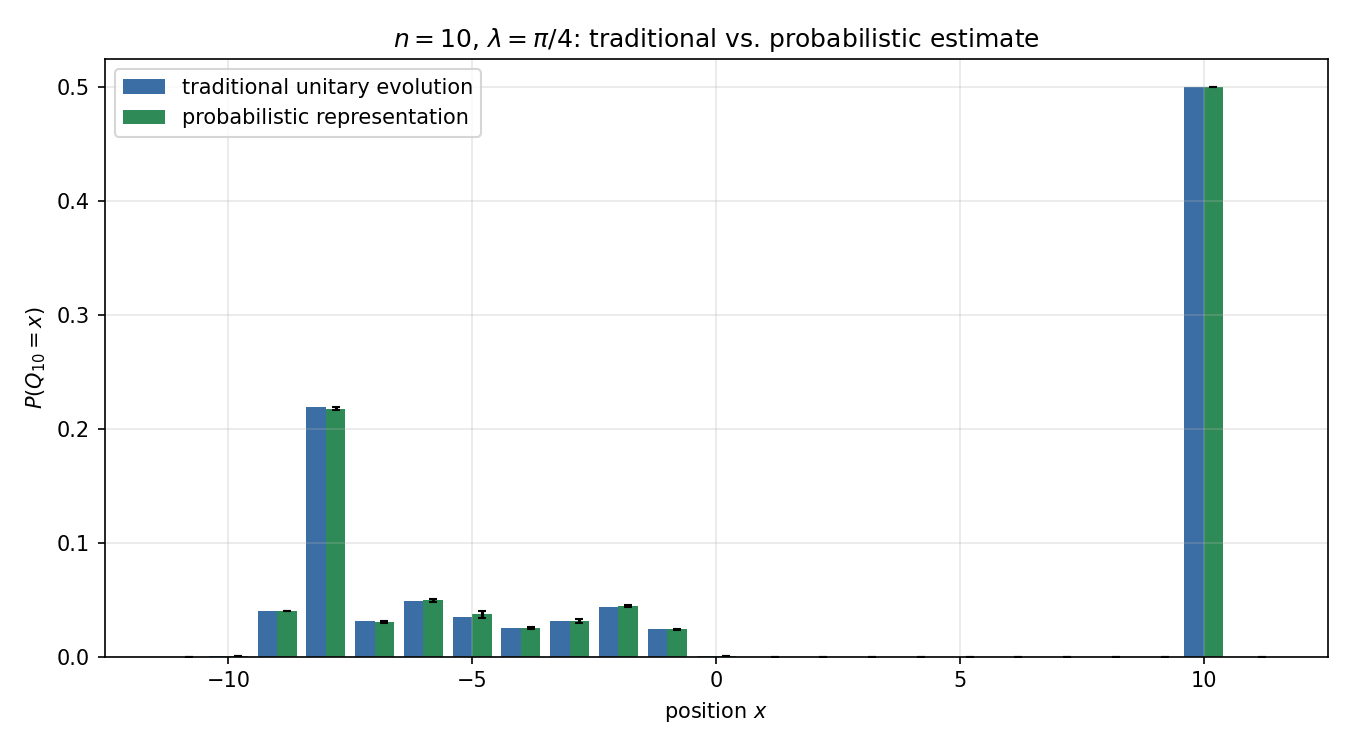, width=12cm}} %100 percent
\vspace*{13pt}
\fcaption{\label{fig4} The three-state quantum walk's probability distribution with the coin $C=e^{i\lambda g_2}$ for $n=10$, $\alpha=\frac{1}{\sqrt{2}}$, $\beta=0$ and $\gamma=\frac{1}{\sqrt{2}}i$, $M=5\times10^6$,and $\lambda=\frac{\pi}{4}$.  }
\end{figure}

Figure \ref{fig4} shows a comparison between the traditional approach of unitary evolution and the probabilistic representation, and confirms the validity of the formulas. \\

Next, consider the Grover walk which is associated with the Grover coin $C=e^{i\frac{7\pi}{4}g_2} e^{i\arccos(-1/3)g_5} e^{i\frac{5\pi}{4}g_2}$ in Equation \eqref{eqg} with the initial state
\begin{align*}
     \ket{\Psi_0}=\ket{0}\otimes\Big(\frac{1}{\sqrt{3}}\ket{1}+\frac{1}{\sqrt{3}}\ket{0}+\frac{1}{\sqrt{3}}\ket{-1}\Big).
\end{align*}

It is well-known that the Grover walk with this initial condition will have the localization at the origin. Figure \ref{fig5} and \ref{fig6} confirm this fact and again the validity of our general coin formula. 

\begin{figure}[H]
\vspace*{13pt}
\centerline{\epsfig{file=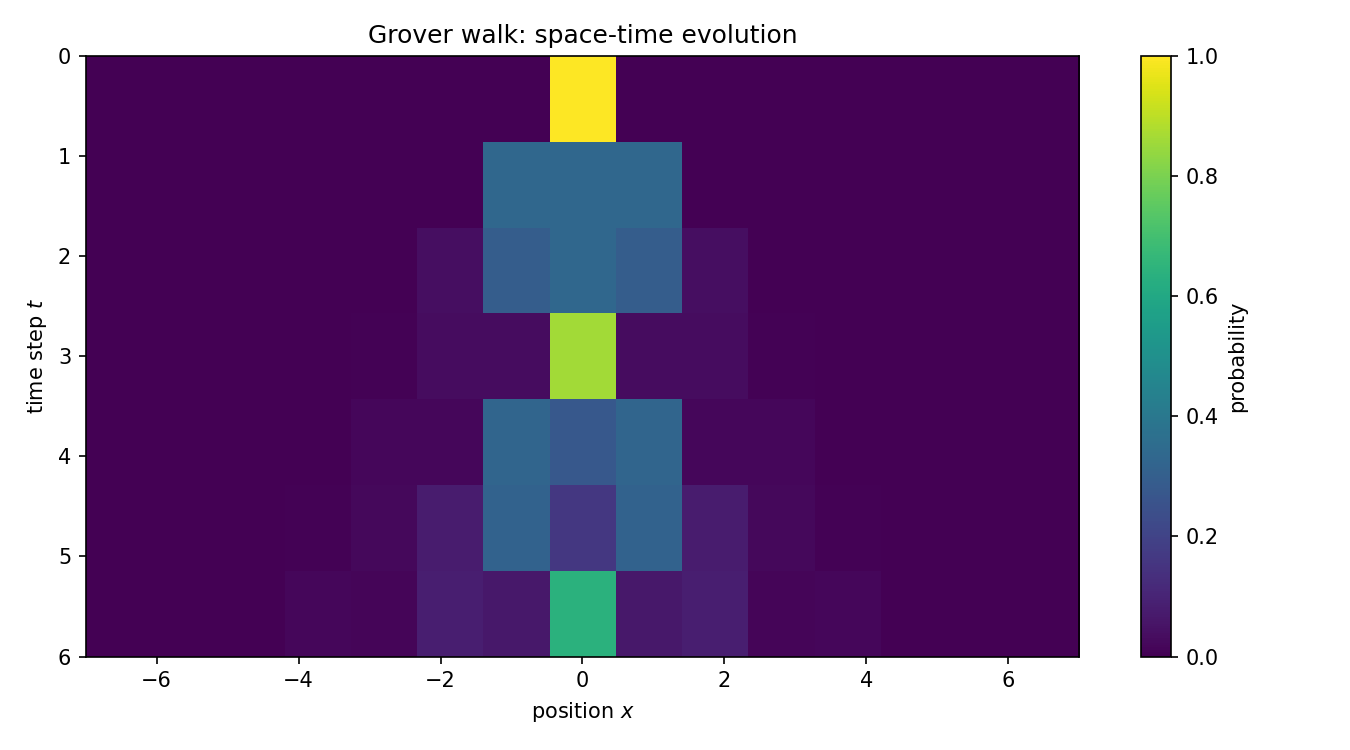, width=10cm}} %100 percent
\vspace*{13pt}
\fcaption{\label{fig5} The three-state quantum walk's spacetime behavior with the Grover coin $C=e^{i\frac{7\pi}{4}g_2} e^{i\arccos(-1/3)g_5} e^{i\frac{5\pi}{4}g_2}$ for $n=6$, $\alpha=\frac{1}{\sqrt{3}}$, $\beta=\frac{1}{\sqrt{3}}$ and $\gamma=\frac{1}{\sqrt{3}}$.  }
\end{figure}

\begin{figure}[H]
\vspace*{13pt}
\centerline{\epsfig{file=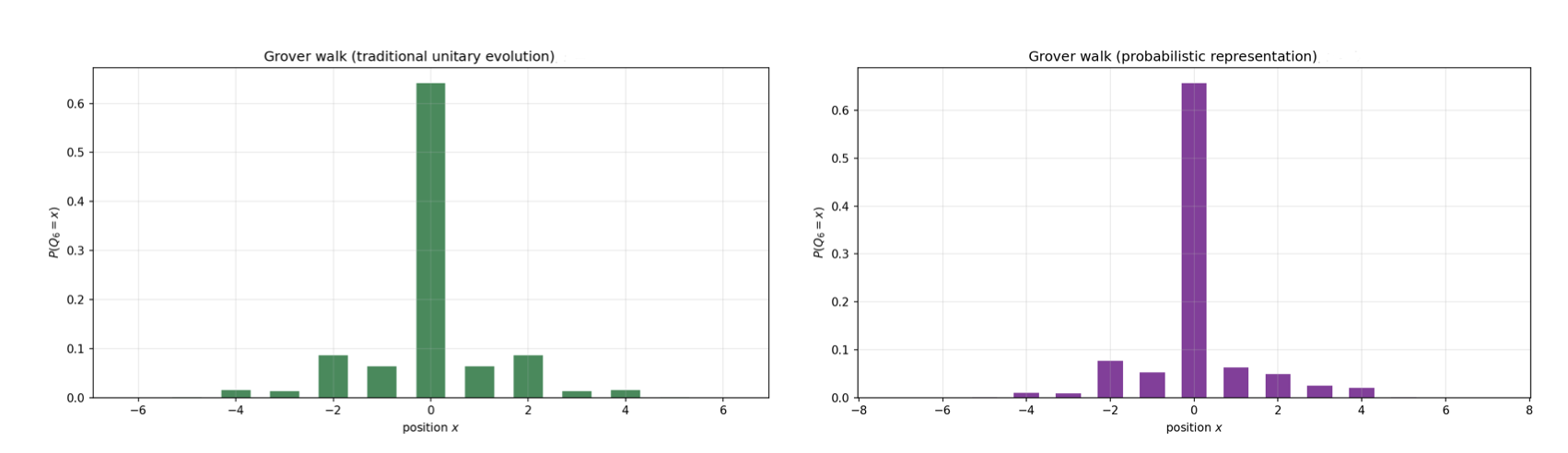, width=18cm}} %100 percent
\vspace*{13pt}
\fcaption{\label{fig6} The Grover walk's probability distribution for $n=6$, $\alpha=\frac{1}{\sqrt{3}}$, $\beta=\frac{1}{\sqrt{3}}$, and $\gamma=\frac{1}{\sqrt{3}}$ with the left bar chart illustrating the benchmark method, and the right bar chart illustrating the probabilistic method with the number of iteration $M=5\times10^8$,$\lambda_0=\lambda_1=\lambda_3=\lambda_5=\lambda_7=\lambda_8=0$, $\lambda_2=\frac{7\pi}{4}$, $\lambda_4=\arccos(-1/3)$, and $\lambda_6=\frac{5\pi}{4}$.  }
\end{figure}

\section{A Convergence to Multi-states Dirac's PDEs System}{\label{sec4}}
In this section, we illustrate how one can utilize the probabilistic representation to obtain the solution to the system of Dirac's PDEs. It is well-known that the discrete-time quantum walk converges to the solution of the Dirac's PDEs (see e.g. Suzuki \& Maeda (2020)\cite{mae}), and with the probabilistic representation we derived in Section \ref{sec3}, one can have a similar result. We will rescale the discrete time $n\in \mathbb N$ and the discrete space $x\in \mathbb Z$ to obtain the a limiting PDE in continuous time $t>0$ and space $z\in\mathbb R$. For simplicity, we consider the three-state discrete-time quantum walk in Definition \ref{def3} with coin $C=e^{i\lambda g_2}$. We will first introduce the following lemma
\begin{lemma}{\label{lm1}}
    Fix $y\in\{-1,0\}$ and we define the flip-count process $B_n$ such that
    $$B_n:=\#\{1\leq j\leq n:N_j \text{ is odd}\}.$$
    Then, we have the following properties
    \begin{align}
      (i) \quad   B_n\equiv S_n (\bmod 2 ) \text{ for every } n.
    \end{align}
    \begin{align}
         (ii)\quad f_n=2B_n+(B_n \bmod 2)(2y+1).
    \end{align}
    Let $y^*:=-1-y$, and set $h(k,-1):=(-1)^{\lceil k/2\rceil}$, $h(k,0):=(-1)^{\lfloor k/2\rfloor}$, $h(k,1):=\mathbb I_{k=0}$ for integer $k\geq0$, we have
    \begin{align}
        (iii) \quad  i^{S_n+f_n}=h(B_n,y) \quad \quad \text{ on the event } \{S_n=B_n\}.
    \end{align}
\end{lemma}
\begin{proof}
    By definition \ref{def1}, $S_n \bmod 2=\sum_{j\leq n}(N_j \bmod 2)\bmod 2= B_n \bmod 2$ because when $N_j$ is odd it contributes exactly $1$ to the parity and $0$ when it is even. This gives us $(i)$. \\
    
    For $(ii)$, notice that $Y_0=y$, and $Y_j\neq Y_{j-1}$ exactly when the parity of $S_j$ differs from that of $S_{j-1}$, thus, the value of $Y_j$ alternates between $y$ and $y^*$ at the flip times $j_1<...<j_{B_n}\leq n$ where $N_j$ is odd, starting from $Y_{j_1}=y^*$. Here, only these flip times contribute to $f_n$ so
    $$f_n=\sum_{i=1}^{B_n}(2Y_{j_i}+3),$$ where 
    \begin{align*}
    Y_{j_i}=
        \begin{cases}
            y \quad \quad \text{ if } i \text{ odd}\\
            y^*    \quad \quad \text{ if } i \text{ even}
        \end{cases}
    \end{align*}
    Let $B_n=2q_1+q_2$, where $q_2=B_n\bmod 2$, so that 
    \begin{align*}
        f_n&=(q_1+q_2)(2y+3)+q_1(2y^*+3)\\
        &=q_1(2y^*+2y+6)+q_2(2y+3)\\
        &=4q_1+q_2(2y+3)\\
        &=2(B_n-q_2)+q_2(2y+3)\\
        &=2B_n+(B_n \bmod 2)(2y+1)
    \end{align*}
    
    Finally, to get $(iii)$, on the event $\{S_n=B_n\}$, using $(ii)$, we have
    \begin{align*}
        S_n+f_n&=B_n+2B_n+(B_n\bmod 2)(2y+1)\\
        &=3B_n+(B_n\bmod 2)(2y+1)
    \end{align*}
    Then a direct computation for each case $y=0$ and $y=-1$ gives us the claim in $(iii)$.
\end{proof}

Lemma \ref{lm1} leads to the following proposition

\begin{proposition}{\label{pr42}}
    Given any $\epsilon>0$ we define $\{N_n^\epsilon\}$ as $n$ i.i.d Poisson random variables with intensity $\epsilon \lambda>0$ and $S_n^\epsilon:=\sum_{i=1}^nN_i^\epsilon$ with $S_0^\epsilon=0$ for $n=1,2,...$, and $B_n^\epsilon:=\#\{1\leq j\leq n:N_j^{\epsilon} \text{ is odd}\}$ then we have the following weak convergence in the space of cadlag path $D(0,\infty)$
    \begin{align}
        \lim_{\epsilon\to 0}\Big(S_{\lfloor t/\epsilon\rfloor}^\epsilon\Big)_{t\geq 0}=\lim_{\epsilon\to 0}\Big(B_{\lfloor t/\epsilon\rfloor}^\epsilon\Big)_{t\geq 0}=(N_t)_{t\geq 0},\label{eq41}
    \end{align}
    where $N_t$ is a Poisson process with parameter $\lambda>0$.
    \begin{align}
        \lim_{\epsilon\to 0}\Big(Y_{\lfloor t/\epsilon\rfloor}^\epsilon\Big)_{t\geq 0}=\bigg((-1)^{N_t}y+\frac{a_{0,1}(y)}{2}\big((-1)^{N_t}-1\big)\bigg)_{t\geq 0}=:(Y_t)_{t\geq 0},\label{eq42}
    \end{align}
    where $Y_{n}^\epsilon=(-1)^{S_n^\epsilon }y+\frac{a_{0,1}(y)}{2}\big((-1)^{S_n^\epsilon}-1\big)$.
    \begin{align}
        \lim_{\epsilon\to 0}\Big(\epsilon X_{\lfloor t/\epsilon\rfloor}^\epsilon\Big)_{t\geq 0}=\bigg(z-\int_{0}^tY_sds\bigg)_{t\geq 0}=:(X_t)_{t\geq 0},\label{eq43}
    \end{align}
    where $X_{n}^\epsilon=\lfloor z/\epsilon\rfloor-\sum_{i=0}^{n-1}Y_i^\epsilon$.
\end{proposition}
\begin{proof}
    We first check the limiting distribution of an increment $S^\epsilon_{\lfloor t/\epsilon \rfloor} - S^\epsilon_{\lfloor s/\epsilon \rfloor}$ follows the Poisson distribution with parameter $\lambda (t-s)$ for $t>s\geq 0$. Indeed, it is the sum of i.i.d Poisson random variables and we have
    $$\lim_{\epsilon \to 0} \left( \lfloor t/\epsilon \rfloor - \lfloor s/\epsilon \rfloor \right) \epsilon \lambda = \lambda(t - s)$$
    because with floor function gives
    $$t/\epsilon - 1 < \lfloor t/\epsilon \rfloor \le t/\epsilon,$$
    $$t - \epsilon < \epsilon \lfloor t/\epsilon \rfloor \le t,$$
    and as $\epsilon \to 0$ the Squeeze Theorem dictates that $\lim_{\epsilon \to 0} \epsilon \lfloor t/\epsilon \rfloor = t$, and similarly $\lim_{\epsilon \to 0} \epsilon \lfloor s/\epsilon \rfloor = s$. Also, the independence of increments follows from the independence of the random variable $\{N_n^\epsilon\}$. Hence, we obtain the first convergence in Equation \eqref{eq41}.\\

    Next, to see why $B_n^\epsilon$ has the same limit as $S_n^\epsilon$, notice that for each $j$ we have
    \begin{align*}
        0\leq N_j^\epsilon-\mathbb I_{N_j^{\epsilon} \bmod 2=1}\leq N_j^\epsilon\cdot\mathbb I_{N_j^\epsilon\geq 2}.
    \end{align*}
    For $N\sim \text{Poi}(\mu)$, we have 
    \begin{align*}
        \mathbb E[N\cdot \mathbb I_{N\geq 2}]=\mu(1-e^{-\mu})=O(\mu^2) \quad \text{as }\mu \to 0.
    \end{align*}
    Replace $\mu=\epsilon \lambda$, and summing over $n=\lfloor t/\epsilon\rfloor$ i.i.d terms,
    \begin{align*}
        \mathbb E[S_{\lfloor t/\epsilon\rfloor}^\epsilon-B_{\lfloor t/\epsilon\rfloor}^\epsilon]=\lfloor t/\epsilon\rfloor\cdot O(\epsilon^2)=O(\epsilon)\longrightarrow 0 \quad \text{as } \epsilon\to 0.
    \end{align*}
    Since $S_{\lfloor t/\epsilon\rfloor}^\epsilon-B_{\lfloor t/\epsilon\rfloor}^\epsilon$ is a nonnegative, non-decreasing sequence, Markov's inequality gives
    \begin{align*}
        \mathbb P(S_{\lfloor t/\epsilon\rfloor}^\epsilon\neq B_{\lfloor t/\epsilon\rfloor}^\epsilon) \to 0 \quad \text{ for each fixed } t. 
    \end{align*}
    This gives the full convergence in Equation \eqref{eq41}.\\
    
    The convergence in Equation \eqref{eq42} follows from the convergence in distribution of integer-valued random variables. Recall that  $Y_{n}^\epsilon=(-1)^{S_n^\epsilon }y+\frac{a_{0,1}(y)}{2}\big((-1)^{S_n^\epsilon}-1\big)$ meaning that the value of $Y_{n}^\epsilon$ at any time $n$ depends only on whether the integer-valued random variable $S_n^\epsilon$ is even or odd, and from Equation \eqref{eq41} we have $$\lim_{\epsilon \to 0} \mathbb P\left( S_{\lfloor t/\epsilon \rfloor}^\epsilon = k \right) = \mathbb P(N_t = k),$$ for every integer $k$. As the mapping is preserved and only takes value $y$ or $y^*$ based on the integer input, the probability that the rescaled direction process is in a certain state is simply the sum of probabilities of the corresponding integers
    $$\mathbb P\left( Y_{\lfloor t/\epsilon \rfloor}^\epsilon = y \right) = \mathbb P\left( S_{\lfloor t/\epsilon \rfloor}^\epsilon \text{ is even} \right),$$
    $$\mathbb P\left( Y_{\lfloor t/\epsilon \rfloor}^\epsilon = y^* \right) = \mathbb P\left( S_{\lfloor t/\epsilon \rfloor}^\epsilon \text{ is odd} \right).$$ Since the probability of every individual integer $k$ converges, the sum of probabilities for all even integers and all odd integers also converges.\\
    
    For the convergence in Equation \eqref{eq43}, we introduce the jump time $n_{i}^\epsilon$ of the process $S_n^\epsilon$, and note that only the increments by 1 matter as the increments by more than 1 are of higher order in $\epsilon$. It follows that $\lim_{\epsilon \to 0} \epsilon n^\epsilon_i = t_i,$ where $t_i$ is the jump time of $(N_t)_{t\geq 0}$. The convergence follows as $\lim_{\epsilon \to 0} \epsilon \lfloor z/\epsilon \rfloor = z$ and the decomposition of sum in $\epsilon X_{\lfloor t/\epsilon\rfloor}^\epsilon$ over the interval $[n_i^\epsilon,n_{i+1}^\epsilon]$ to obtain the convergence to the integral appearing in $(X_t)_{t\geq 0}$ over the interval $[t_i,t_{i+1}]$. This completes our proof.
\end{proof}

Now, we have the following theorem
\begin{theorem}{\label{tm43}}
    Given $\epsilon>0$ define the initial state $\Psi_0^\epsilon\in \ell^2(\mathbb Z\times\{0,\pm 1\})$ by $\Psi_0^\epsilon(x,y) :=K(\epsilon x, y)$, where $K:\mathbb R\times \{0,\pm 1\}\to \mathbb C$ with $\lim_{z\to \pm \infty}|K(z,y)|=0$. Then, we have the following pointwise limite
    \begin{align}
        \Psi(t,z,y):=\lim_{\epsilon \to 0}\Psi_{\lfloor t/\epsilon\rfloor}^\epsilon(\lfloor z/\epsilon\rfloor,y)=e^{\lambda t}\mathbb E[h(N_t,y)\cdot K(X_t,Y_t)],
    \end{align}
    where $\Psi_n^\epsilon$ be the $n-$step evolution of applying the homogeneous coin $C=e^{i\epsilon\lambda g_2}$ to the initial state $\Psi_0^\epsilon$ for every $t>0$, $z\in \mathbb R$, and $y=0,\pm 1$ with $N_0=0$, $X_0=z$, and $Y_0=y$. In particular, 
    \begin{align}
        \Psi(t,z,1)=K(z-t,1).\label{eq48}
    \end{align}
\end{theorem}
\begin{proof}
    Notice that for $y=1$ only the sample path with all $N_j^\epsilon=0$ contributes, and then in this case $f_n\equiv 0$ and $\mathbb P(S_n^\epsilon=0)=e^{-n\epsilon\lambda}$, thus, using the probabilistic representation in Equation \eqref{eqg2} we have
    \begin{align*}
        \Psi_n^\epsilon(x,1)=e^{n\epsilon\lambda}\cdot e^{-n\epsilon\lambda}\cdot \Psi_0(x-n,1)=K(\epsilon(x-n),1).
    \end{align*}
    Rescale using Proposition \ref{pr42} gives $\Psi(t,z,1)=K(z-t,1).$\\
    For $y\in\{0,-1\}$, the probabilistic representation in Equation \eqref{eqg2} gives
    \begin{align*}
        \Psi_n^\epsilon(x,y)=e^{n\epsilon\lambda}\cdot\mathbb E[i^{S_n^\epsilon+f_n}K(\epsilon X_n^\epsilon,Y_n^\epsilon)].
    \end{align*}
    By Lemma \ref{lm1}, $i^{S_n^\epsilon+f_n}=h(B_n^\epsilon,y)$ on the event $\{S_n^\epsilon=B_n^\epsilon\}$, an event whose probability converges to $1$ due to Proposition \ref{pr42}. Since $|i^{S_n^\epsilon+f_n}|=|h(.,.)|=1$, the discrepancy on the complementary event is bounded and vanishing in probability, contributing $o(1)$ to the expectation. On the event $\{S_n^\epsilon=B_n^\epsilon\}$, replace $n=\lfloor t/\epsilon\rfloor$, $x=\lfloor z/\epsilon\rfloor$, and use Proposition \ref{pr42} to obtain 
    \begin{align*}
        \Psi(t,z,y)=e^{\lambda t}\mathbb E[h(N_t,y)\cdot K(X_t,Y_t)],
    \end{align*}
    where we can pass the limit inside the expectation due to bounded convergence theorem.
\end{proof}

On the other hand, we have
\begin{lemma}{\label{lm2}}
    For $y\in\{-1,0\}$, set $\Phi(y,k):=(-1)^ky+\frac{a_{0,1}(y)}{2}((-1)^k-1)$ so that $Y_t=\Phi(y,N_t)$, then $\Phi(y,k)\in\{-1,0\}$ for every $k\geq 0$, and for all integer $a,b\geq 0$
    \begin{align*}
        \Phi(y,a+b)=\Phi(\Phi(y,a),b).
    \end{align*}
\end{lemma}
\begin{proof}
    For $k$ even, $\Phi(y,k)=y \in \{-1,0\}$, and for $k$ odd, $\Phi(y,k)=y^*\in\{-1,0\}$.\\
    Next, we have
    \begin{align*}
         \Phi(\Phi(y,a),b)&=(-1)^b\Phi(y,a)+\frac{1}{2}((-1)^b-1)\\
    &=(-1)^{a+b}y+\frac{1}{2}((-1)^{a+b}-1).
    \end{align*}
   Also, by definition,
   \begin{align*}
       \Phi(y,a+b)=(-1)^{a+b}y+\frac{1}{2}((-1)^{a+b}-1).
   \end{align*}
   Hence,
   \begin{align*}
        \Phi(y,a+b)=\Phi(\Phi(y,a),b).
    \end{align*}
    This completes our proof.
\end{proof}

The Theorem \ref{tm43} and Lemma \ref{lm2} gives rise to the following Lemma
\begin{lemma}{\label{lm4}}
    The limit $\Psi(t,z,y)$ in Theorem \ref{tm43} solves the linear Dirac's PDEs system
    \begin{align}
        \begin{cases}
            \partial_t\Psi(t,z,1)=-\partial_z \Psi(t,z,1)\\
            \partial_t\Psi(t,z,0)=\lambda \Psi(t,z,-1) \\
            \partial_t\Psi(t,z,-1)= \partial_z \Psi(t,z,-1)-\lambda  \Psi(t,z,0)
        \end{cases},
    \end{align}
    with $\Psi(0,z,y)=K(z,y)$, and probability is preserved
    \begin{align*}
        \int_{-\infty}^\infty\bigg(|\Psi(t,z,1)|^2+|\Psi(t,z,0)|^2+|\Psi(t,z,-1)|^2\bigg)dz=1 \quad \quad \text{ if it holds at } t=0.
    \end{align*}
\end{lemma}
\begin{proof}
For $y=1$, Equation \eqref{eq48} gives directly $\partial_t\Psi(t,z,1)=-\partial_z \Psi(t,z,1)$.\\ 

Let $\mathcal F_{\Delta_t}:=\sigma(N_s:s\leq \Delta_t)$, and set $I_t:=N_{t+\Delta_t}-N_{\Delta_t}$. Here, $I_t$ is independent of the filtration $\mathcal F_{\Delta_t}$, and has the same distribution as $N_t$.\\

For $y\in\{0,-1\}$, we have
\begin{align*}
    \Psi(t+\Delta_t, z, y)&=e^{\lambda (t+\Delta_t)} \mathbb E_{0,z,y} \bigg[ h({N_{t+\Delta_t}},y) \cdot K\bigg( z -  \int_0^{t+\Delta_t} \Phi(y,N_s) \, ds,
    \Phi(y,N_{t+\Delta_t} \bigg) \bigg]\\
    &=e^{\lambda (t+\Delta_t)} \mathbb E_{0,z,y} \bigg[h(N_{\Delta_t},y)
    \mathbb E\bigg[h(I_t,Y_{\Delta_t})\cdot K\bigg( X_{\Delta_t}- \int_{0}^{t} \Phi(y,N_{\Delta_t}+I_u) \, du,\\\Phi(y,N_{\Delta_t}+I_t)
     \bigg)\bigg|\mathcal F_{\Delta_t}\bigg]\bigg]\\
     &=e^{\lambda (t+\Delta_t)} \mathbb E_{0,z,y} \bigg[h(N_{\Delta_t},y)
    \mathbb E\bigg[h(I_t,Y_{\Delta_t})\cdot K\bigg( X_{\Delta_t}- \int_{0}^{t} \Phi(Y_{\Delta_t},I_u) \, du,\\\Phi(Y_{\Delta_t},I_t)
     \bigg)\bigg|\mathcal F_{\Delta_t}\bigg]\bigg]\\
     &=e^{\lambda (t+\Delta_t)} \mathbb E_{0,z,y} \bigg[h(N_{\Delta_t},y)
    \mathbb E\bigg[h(I_t,Y_{\Delta_t})\cdot K\bigg( X_{\Delta_t}- \int_{0}^{t} \Phi(Y_{\Delta_t},I_u) \, du,\\\Phi(Y_{\Delta_t},I_t)
     \bigg)\bigg]\bigg]\\
     &= e^{\lambda (t+\Delta_t)} \mathbb E_{0,z,y} \bigg[h(N_{\Delta_t},y)\cdot e^{-\lambda t}\Psi(t,X_{\Delta_t}, Y_{\Delta_t})\bigg]\\
     &= e^{\lambda\Delta_t}\mathbb E_{0,z,y} \bigg[h(N_{\Delta_t},y)\Psi(t,X_{\Delta_t}, Y_{\Delta_t})\bigg],
\end{align*}
where we use the Tower property in line 2, Lemma \ref{lm2} in line 3, and the independence and measurability of $(X_{\Delta_t}, Y_{\Delta_t},I_t)$ with respect to the filtration $\mathcal F_{\Delta_t}$ in line 4.\\

Now, we have
\begin{align*}
    \Psi(t+\Delta_t, z, y)= e^{\lambda \Delta t} \mathbb P(N_{\Delta t} = 0)\cdot\mathbb  E_{0,z,y} [ h(0,y) \Psi(t, X_{\Delta t}, Y_{\Delta t}) \mid N_{\Delta t} = 0 ]\\
    +e^{\lambda \Delta_t} \mathbb P(N_{\Delta_t} = 1)\cdot \mathbb E_{0,z,y} \left[ h(1,y) \Psi(t, X_{\Delta_t}, Y_{\Delta_t}) \mid N_{\Delta_t} = 1 \right] \\+ O((\Delta_t)^2).
\end{align*}

That is
\begin{align*}
    \Psi(t+\Delta_t, z, y)&=\Psi(t,z-y\Delta_t,y)+\lambda\Delta_t h(1,y)\mathbb E_\tau[\Psi(t,z-y\tau-y^*(\Delta_t-\tau ),y^*)]+O(\Delta_t^2)\\
    &=\Psi(t,z-y\Delta_t,y)+\lambda h(1,y)\int_0^{\Delta_t}\Psi(t,z-y\tau-y^*(\Delta_t-\tau ),y^*)d\tau+O(\Delta_t^2),
\end{align*}
where $h(0,y)=1$, $h(1,y)=1$ if $y=0$, $h(1,y)=-1$ if $y=-1$, $\tau$ is the first jump time, and $\tau \sim U(0,\Delta_t)$.\\

Subtracting $\Psi(t,z,y)$ from both sides, dividing by $\Delta_t$ and let $\Delta_t\to 0$ we can get
\begin{align*}
    \partial_t \Psi(t,z,y)=-y\partial_z\Psi(t,z,y)+\lambda h(1,y)\Psi(t,z,y^*).
\end{align*}
For the conservation law, let $u=\Psi(t,z,0)$ and $v=\Psi(t,z,-1)$. Then,
\begin{align*}
    \frac{d}{dt}\int (|u|^2+|v|^2)dz=\int 2\text{ Re}(\bar{u}\partial_tu+\bar{v}\partial_t v)dz=\int 2\text{ Re}(-\lambda\bar{u}v+\bar{v}\partial_z v+\lambda \bar{v}u)dz,
\end{align*}
where the two $\lambda-$terms combine to $\lambda(u\bar{v}-\bar{u}v)=2i\lambda \text{ Im}(u\bar{v})$, thus their real parts vanish identically. What remains is 
\begin{align*}
    \int 2\text{ Re}(\bar{v}\partial_z v)dz=\int\partial_z|v|^2dz=0,
\end{align*}
by the decay of $v$ at $\pm\infty$. Hence, $\int(|u|^2+|v|^2)dz$ is constant in $t$, together with the trivially norm-preserving of $\Psi(t,z,1)$, gives the stated conservation law.
\end{proof}
In the case of the general coin $C=e^{i\lambda_0}e^{i\lambda_1g_3}e^{i\lambda_2g_2}e^{i\lambda_3g_3}e^{i\lambda_4g_5}e^{i\lambda_5g_3}e^{i\lambda_6g_2}e^{i\lambda_7g_3}e^{i\lambda_8g_8}$, the generalization is straightforward. Let  $\Lambda=\lambda_1+\lambda_3+\lambda_5+\lambda_7$, and $\Theta(y):=\Lambda\theta_3(y)+\lambda_8\theta_8(y)$, we have
\begin{lemma}
    The limit $\Psi(t,z,y)$ for the general coin $C$ in Equation \eqref{eq0} solves the linear Dirac's PDEs system
    \begin{align}
        \begin{cases}
            \partial_t\Psi(t,z,1)=-\partial_z\Psi(t,z,1)+i\Theta(1)\Psi(t,z,1)+\lambda_4\Psi(t,z,-1)\\
            \partial_t\Psi(t,z,0)=i\Theta(0)\Psi(t,z,0)+(\lambda_2+\lambda_6)\Psi(t,z,-1)\\
            \partial_t\Psi(t,z,-1)=\partial_z\Psi(t,z,-1)+i\Theta(-1)\Psi(t,z,-1)-(\lambda_2+\lambda_6)\Psi(t,z,0)-\lambda_4\Psi(t,z,1)
        \end{cases},
    \end{align}
    and the probability is preserved.
\end{lemma}
\begin{proof}
    Similar to proof of Lemma \ref{lm4}. 
\end{proof}
\section{Conclusion}

In conclusion, this work investigates the connection between classical stochastic theory and quantum dynamics by constructing a rigorous probabilistic formulation for multi-state discrete-time quantum walks on integer lattices. Although quantum walks are governed by deterministic principles, our findings show that framing them within a probabilistic context uncovers a much stronger link to classical random processes than has been traditionally recognized.\\

In Section \ref{sec3}, we successfully expanded the foundational framework established by Vu (2026) \cite{vu} from two-state systems to three-state walks—a non-trivial generalization owing to the presence of annihilation phenomena in higher-dimensional state spaces. By testing our formulation on the Gell-Mann coin matrix $g_2$, we confirmed that our probabilistic equations reliably reproduce known quantum dynamics, offering an efficient computational surrogate to conventional unitary propagation strategies.\\

Reinterpreting these systems through probability rather than purely functional analysis paves the way for several compelling future directions: our representation offers a feasible method to bypass the analytical obstacles associated with multi-dimensional quantum walks, where establishing weak limit theorems remains a major challenge. Furthermore, these expressions establish a framework for implementing variance-reduction strategies and classical Monte Carlo schemes in quantum settings. Ultimately, projecting quantum probability amplitudes onto classical stochastic models allows researchers to systematically isolate and measure the distinct "quantumness" that sets these walks apart from classical diffusion.

\nonumsection{References}

\begin{appendices}
    
\section{Proof of Lemma 3.1.4}\label{a31}
First, observe that 
      \begin{align*}
        U\ket{x}\ket{y}&=S\cdot(I\otimes C)\ket{x}\ket{y}\\
        &=S\ket{x}e^{i\lambda g_3}\ket{y}\\
        &=\sum_{k\in \mathbb N}S\ket{x}\frac{(i\lambda)^k}{k!}g_3^k\ket{y}\\
        &=\sum_{k\in\mathbb N}\frac{(i\lambda)^k}{k!}a^k_0(y)(-1)^{k(y+1)}\ket{x+y}\ket{y}\\
        &=\sum_{k\in\mathbb N}i^{k}\frac{\lambda^k}{k!}a_0^k(y)(-1)^{k(y+1)}\ket{x+y}\ket{y},
    \end{align*}   
     where $a_0(y)=1-\mathbb I_{y=1}$.\\
     
    Now, for any state $\Psi$ of the walk, we have
     \begin{align*}
        U\Psi&=U\sum_{\substack{x\in\mathbb Z\\y\in\{0,\pm 1\}}}\Psi(x,y)\ket{x}\ket{y}\\
        &=\sum_{\substack{x\in\mathbb Z\\y\in\{0,\pm 1\}\\k\in \mathbb N}}\Psi(x,y)i^{k}\frac{\lambda^k}{k!}a_0^k(y)(-1)^{k(y+1)}\ket{x+y}\ket{y}\\
        &=\sum_{\substack{x\in\mathbb Z\\y\in\{0,\pm 1\}\\k\in \mathbb N}}i^{k}\frac{\lambda^k}{k!}a_0^k(y)(-1)^{k(y+1)}\Psi(x-y,y)\ket{x}\ket{y}
    \end{align*}   
   
   This implies that 
   \begin{align}
       (U\Psi)(x,y)=\sum_{k\in\mathbb N}i^{k}\frac{\lambda^k}{k!}a_0^k(y)(-1)^{k(y+1)}\Psi(x-y,y)
   \end{align}

    Hence, the evolution after $n-$steps yields the probability amplitude
    \begin{align}
       \Psi_n(x,y)=\sum_{k_1,k_2,...,k_n\in \mathbb N}i^{\sum_{j=1}^nk_j}\frac{\lambda^{\sum_{j=1}^nk_j}}{k_1!k_2!...k_n!}a_0^{\sum_{j=1}^nk_j}(y_0)(-1)^{(y_0+1)\sum_{j=1}^nk_j}\Psi(x-ny_0,y_0).
   \end{align}
    This completes our proof. 
    
\section{Proof of Theorem 3.1.5}\label{a32}
 From Equation \eqref{eq2} in Lemma \ref{lem3}, apply the Poisson distribution, we have
    \begin{align*}
         \Psi_n(x_0,y_0)&=\sum_{k_1,...,k_n\in\mathbb N}i^{\sum_{j=1}^n k_j}\frac{\lambda^{\sum_{j=1}^nk_j}}{k_1!...k_n!}a_0^{\sum_{j=1}^nk_j}(y_0)(-1)^{(y_0+1)\sum_{j=1}^nk_j}\Psi_0(x_n,y_n)\\
         &=e^{n\lambda}\sum_{k_1,...,k_n\in\mathbb N}i^{\sum_{j=1}^n k_j}\frac{e^{-\lambda}\lambda^{k_1}...e^{-\lambda}\lambda^{k_n}}{k_1!...k_n!}a_0^{\sum_{j=1}^nk_j}(y_0)(-1)^{(y_0+1)\sum_{j=1}^nk_j}\Psi_0(x_n,y_n)\\
         &=e^{n\lambda}\Psi_0(x-ny_0,y_0)\mathbb E\Big [i^{S_n}a_{0,1}^{S_n}(y_0)(-1)^{S_n(y_0+1)}\Big],
    \end{align*}
    for $x_0=x$, and $y_0=y$. This completes our proof.
\section{Proof of Lemma 3.1.6}\label{a33}
  First observe that 
      \begin{align*}
        U\ket{x}\ket{y}&=S\cdot(I\otimes C)\ket{x}\ket{y}\\
        &=S\ket{x}e^{i\lambda g_5}\ket{y}\\
        &=\sum_{k\in \mathbb N}S\ket{x}\frac{(i\lambda)^k}{k!}g_5^k\ket{y}\\
        &=\sum_{k\in\mathbb N}\frac{(i\lambda)^k}{k!}a^k_0(y)i^{a_{1,5}(k,y)}\ket{x+(-1)^ky}\ket{(-1)^ky}\\
        &=\sum_{k\in\mathbb N}i^{k+a_{1,5}(k,y)}\frac{\lambda^k}{k!}a_0^k(y)\ket{x+(-1)^ky}\ket{(-1)^ky},
    \end{align*}   
     where $a_0(y)=1-\mathbb I_{y=0}$, and $a_{1,5}(k,y)=\frac{1-(-1)^k}{2}(y+2)$.\\
     
    Now, for any state $\Psi$ of the walk, we have
     \begin{align*}
        U\Psi&=U\sum_{\substack{x\in\mathbb Z\\y\in\{0,\pm 1\}}}\Psi(x,y)\ket{x}\ket{y}\\
        &=\sum_{\substack{x\in\mathbb Z\\y\in\{0,\pm 1\}\\k\in \mathbb N}}\Psi(x,y)i^{k+a_{1,5}(k,y)}\frac{\lambda^k}{k!}a_0^k(y)\ket{x+(-1)^ky}\ket{(-1)^ky}\\
        &=\sum_{\substack{x\in\mathbb Z\\y\in\{0,\pm 1\}\\k\in \mathbb N}}i^{k+a_{1,5}(k,(-1)^ky)}\frac{\lambda^k}{k!}a_0^k(y)\Psi(x-y,(-1)^ky)\ket{x}\ket{y}
    \end{align*}   
   
   This implies that 
   \begin{align}
       (U\Psi)(x,y)=\sum_{k\in\mathbb N}i^{k+a_{1,5}(k,(-1)^ky)}\frac{\lambda^k}{k!}a_0^k(y)\Psi(x-y,(-1)^ky)
   \end{align}
  
    Hence, the evolution after $n-$steps yields the probability amplitude
    \begin{align}
       \Psi_n(x,y)=\sum_{k_1,k_2,...,k_n\in \mathbb N}i^{\sum_{j=1}^nk_j+a_{1,5}(k_j,y_{j-1})}\frac{\lambda^{\sum_{j=1}^nk_j}}{k_1!k_2!...k_n!}a_0^{\sum_{j=1}^nk_j}(y_0)\Psi(x-y_0-\sum_{j=1}^{n-1}(-1)^{k_j}y_j,(-1)^{k_n}y_{n-1}).
   \end{align}
    This completes our proof. 
    
\section{Proof of Theorem 3.1.7}\label{a34}
 From Equation \eqref{eq3} in Lemma \ref{lem4}, apply the Poisson distribution, we have
    \begin{align*}
         \Psi_n(x_0,y_0)&=\sum_{k_1,...,k_n\in\mathbb N}i^{\sum_{j=1}^n k_j+\frac{(1-(-1)^{k_j})(y_{j-1}+2)}{2}}\frac{\lambda^{\sum_{j=1}^nk_j}}{k_1!...k_n!}a_0^{\sum_{j=1}^nk_j}(y_0)\Psi_0(x_n,y_n)\\
         &=e^{n\lambda}\sum_{k_1,...,k_n\in\mathbb N}i^{\sum_{j=1}^n k_j+\frac{(1-(-1)^{k_j})(y_{j-1}+2)}{2}}\frac{e^{-\lambda}\lambda^{k_1}...e^{-\lambda}\lambda^{k_n}}{k_1!...k_n!}a_0^{\sum_{j=1}^nk_j}(y_0)\Psi_0(x_n,y_n)\\
         &=e^{n\lambda}\mathbb E\Big [i^{S_n+r_n}a_{0,1}^{S_n}(Y_0)\Psi_0(X_n,Y_n)\Big|(S_0,X_0,Y_0)=(0,x,y)\Big],
    \end{align*}
    for $x_0=x$, $y_0=y$, and $r_n=\sum_{j=1}^n\frac{(1-(-1)^{k_j})(y_{j-1}+2)}{2}$. This completes our proof.
\section{Proof of Lemma 3.1.8}\label{a35}
 First observe that 
      \begin{align*}
        U\ket{x}\ket{y}&=S\cdot(I\otimes C)\ket{x}\ket{y}\\
        &=S\ket{x}e^{i\lambda g_8}\ket{y}\\
        &=\sum_{k\in \mathbb N}S\ket{x}\frac{(i\lambda)^k}{k!}g_8^k\ket{y}\\
        &=\sum_{k\in\mathbb N}\frac{(i\lambda)^k}{k!}b_0^k(y)\ket{x+y}\ket{y}\\
        &=\sum_{k\in\mathbb N}i^{k}\frac{\lambda^k}{k!}b_0^k(y)\ket{x+y}\ket{y},
    \end{align*}   
     where $b_0(y)=\frac{(-1)^{\mathbb I_{y=1}}(\mathbb I_{y=1}+1)}{\sqrt{3}}$.\\
     
    Now, for any state $\Psi$ of the walk, we have:
     \begin{align*}
        U\Psi&=U\sum_{\substack{x\in\mathbb Z\\y\in\{0,\pm 1\}}}\Psi(x,y)\ket{x}\ket{y}\\
        &=\sum_{\substack{x\in\mathbb Z\\y\in\{0,\pm 1\}\\k\in \mathbb N}}\Psi(x,y)i^{k}\frac{\lambda^k}{k!}b_0^k(y)\ket{x+y}\ket{y}\\
        &=\sum_{\substack{x\in\mathbb Z\\y\in\{0,\pm 1\}\\k\in \mathbb N}}i^{k}\frac{\lambda^k}{k!}b_0^k(y)\Psi(x-y,y)\ket{x}\ket{y}
    \end{align*}   
   
   This implies that 
   \begin{align}
       (U\Psi)(x,y)=\sum_{k\in\mathbb N}i^{k}\frac{\lambda^k}{k!}b_0^k(y)\Psi(x-y,y)
   \end{align}

    Hence, the evolution after $n-$steps yields the probability amplitude:
    \begin{align}
       \Psi_n(x,y)=\sum_{k_1,k_2,...,k_n\in \mathbb N}i^{\sum_{j=1}^nk_j}\frac{\lambda^{\sum_{j=1}^nk_j}}{k_1!k_2!...k_n!}b_0^{\sum_{j=1}^nk_j}(y_0)\Psi(x-ny_0,y_0).
   \end{align}
    This completes our proof. 
    
\section{Proof of Theorem 3.1.9}\label{a36}
From Equation \eqref{eq4} in Lemma \ref{lem5}, apply the Poisson distribution, we have
    \begin{align*}
         \Psi_n(x_0,y_0)&=\sum_{k_1,...,k_n\in\mathbb N}i^{\sum_{j=1}^n k_j}\frac{\lambda^{\sum_{j=1}^nk_j}}{k_1!...k_n!}b_0^{\sum_{j=1}^nk_j}(y_0)\Psi_0(x_n,y_n)\\
         &=e^{n\lambda}\sum_{k_1,...,k_n\in\mathbb N}i^{\sum_{j=1}^n k_j}\frac{e^{-\lambda}\lambda^{k_1}...e^{-\lambda}\lambda^{k_n}}{k_1!...k_n!}a_0^{\sum_{j=1}^nk_j}(y_0)\Psi_0(x_n,y_n)\\
         &=e^{n\lambda}\Psi_0(x-ny_0,y_0)\mathbb E\Big [i^{S_n}b_{0}^{S_n}(y_0)\Big],
    \end{align*}
    for $x_0=x$, and $y_0=y$. This completes our proof.
\end{appendices}

\noindent


\begin{thebibliography}{000}
\bibitem{amb}
Ambainis, A., Bach, E., Nayak, A., Vishwanath, A., and Watrous, J. (2001). One-dimensional quantum
walks, Proc. of the 33rd Annual ACM Symposium on Theory of Computing, 37–49.

\bibitem{carmona}
Ren\'e Carmona. Random Schr\"odinger operators. Ecole d’Ete de Probabilites de Saint Flour XIV, pages 1–124,
1984.

\bibitem{child} 
Childs, A. M., Farhi, E., and Gutmann, S. (2002). An example of the diﬀerence between quantum and
classical random walks, Quantum Information Processing, 1, 35–43, quant-ph/0103020.

\bibitem{ch}
Childs, A. M. (2022). Lecture notes on quantum algorithms. University of Maryland. https://www.cs.umd.edu/~amchilds/qa/qa.pdf

\bibitem{gre}
Greiner, W., and Muller, B. (1989). Quantum Mechanics: Symmetries (Berlin: Springer)

\bibitem{grimmett}
Grimmett, G., Janson, S., and Scudo, P. F. (2004). Weak limits for quantum random walks, Phys. Rev. E,
69, 026119, quant-ph/0309135.

\bibitem{gudder}
Gudder, S. P. (1988). Quantum Probability. Academic Press Inc., CA.

\bibitem{konno}
Konno, N. (2002a). Quantum random walks in one dimension, Quantum Information Processing, 1, 345–354, quant-ph/0206053.

\bibitem{konno2}
Konno, N. (2005a). Limit theorem for continuous-time quantum walk on the line, Phys. Rev. E, 72, 026113, quant-ph/0408140.

\bibitem{konno3}
Konno, N., Matsue, K., and Segawa, E. (2023). A crossover between open quantum random walks to quantum walks. Journal of Statistical Physics, 190(12):202.

\bibitem{mae}
Maeda, M., and Suzuki, A. (2020). Continuous limits of linear and nonlinear quantum walks. Reviews in Mathematical Physics, 32(04), 2050008. https://doi.org/10.1142/S0129055X20500087

\bibitem{meyer}
Meyer, D. A. (1996). From quantum cellular automata to quantum lattice gases, J. Statist. Phys., 85, 551–574, quant-ph/9604003.

\bibitem{mon}
Montero, M. (2017). Quantum and random walks as universal generators of probability distributions. Physical Review A, 95(6):062326.

\bibitem{nayak}
Nayak, A., and Vishwanath, A. (2000). Quantum walk on the line, quant-ph/0010117.

\bibitem{po}
Portugal, R. (2018). Quantum walks and search algorithms (2nd ed.). Springer Nature. https://doi.org/10.1007/978-3-319-97813-0

\bibitem{todd}
Tilma, T., and Sudarshan ECG., (2002). Generalized Euler angle parametrization for SU(N). Journal of Physics A: Mathematical and General. 

\bibitem{vu}
Vu, H. (2026). Molchanov's Formula and Quantum Walks: A Probabilistic Approach. https://arxiv.org/abs/2601.01071

\bibitem{yama}
Yamagami, T., Segawa, E., Chauvet, N., Rohm, A., Horisaki, R., and Naruse, M. (2022). Directivity of quantum walk via its random walk replica. Complexity, 2022(ID 9021583):114.

\end{thebibliography}
\end{document}